\documentclass[ 10pt,journal, twocolumn, final]{IEEEtran}

\usepackage{refcheck}
\usepackage{graphicx}
\usepackage{amsmath}
\usepackage{amssymb}
\usepackage{graphics}
\usepackage{latexsym}
\usepackage[dvips]{epsfig}
\usepackage[nospace]{cite}
\usepackage{tabularx}
\usepackage{subfigure}
\usepackage{epstopdf}
\usepackage{mathdots}

\newtheorem{Lemma}{\bf Lemma}

\newtheorem{Proposition}[Lemma]{\bf Proposition}
\newtheorem{Theorem}{\bf Theorem}

\newtheorem{Remark}{Remark}

\def\Pr{{\rm \mathbf {Pr}}}
\def\E{{\rm \mathbf  E}}

\newcommand{\argmax}{\operatornamewithlimits{argmax}}
\newcommand{\argmin}{\operatornamewithlimits{argmin}}

\allowdisplaybreaks

\begin{document}

\title{Measuring Secrecy by \\
the Probability of a Successful Guess }

\author{
   \IEEEauthorblockN{Ibrahim Issa and Aaron B. Wagner}
    \thanks{The authors are with the School of Electrical and Computer Engineering, Cornell University, Ithaca, NY (email: ii47@cornell.edu, wagner@cornell.edu). This paper was presented in part at the 53rd Annual Allerton Conference on Communication, Control, and Computing (2015).   Copyright (c) 2014 IEEE. Personal use of this material is permitted.  However, permission to use this material for any other purposes must be obtained from the IEEE by sending a request to pubs-permissions@ieee.org.}
   }

\maketitle

\begin{abstract}
The secrecy of a communication system in which both the legitimate receiver and an eavesdropper are allowed some distortion is investigated. The secrecy metric considered is the exponent of the probability that the eavesdropper estimates the source sequence successfully within an acceptable distortion level. 
The problem is first studied when the transmitter and the legitimate receiver do not share any key and the transmitter is not subject to a rate constraint, which corresponds to a stylized model of a side channel and reveals connections to source coding with side information.  The setting is then generalized to include a shared secret key between the transmitter and the legitimate receiver and a rate constraint on the transmitter, which corresponds to the Shannon cipher system. A single-letter characterization of the highest achievable exponent is provided, and asymptotically-optimal strategies for both the primary user and the eavesdropper are demonstrated. 
%
%When the transmitter and the legitimate receiver do not share any key and the transmitter is not subject to a rate constraint, a single-letter characterization of the highest achievable exponent is provided. Moreover, asymptotically-optimal universal strategies for both the primary user and the eavesdropper are demonstrated, where universality means independence of the source statistics. When the transmitter and the legitimate receiver share a secret key and the transmitter is subject to a rate constraint, upper and lower bounds are derived on the exponent by analyzing the performance of suggested strategies for the primary user and the eavesdropper. The bounds admit a single-letter characterization and they match under certain conditions, which include the case in which the eavesdropper must reconstruct the source exactly.
\end{abstract} 

\section{Introduction}

To compromise the security of a communication network, an eavesdropper need not have direct access to the decrypted content of the transmitted packets. In fact, simply monitoring and analyzing the network flow may help an eavesdropper deduce sensitive information. For example, Song \emph{et al.}~\cite{SSHTiming} show that the Secure Shell (SSH) is vulnerable to what is called \emph{timing} attacks. In SSH, each keystroke is immediately sent to the remote machine, and an eavesdropper can thus observe the timing of the keystrokes. It is shown that this information can be used to significantly speed up the search for passwords, and it is estimated that each consecutive pair of keystrokes leaks around 1 bit of information. Zhang and Wang~\cite{PeepingTom} enhance the attack proposed in~\cite{SSHTiming}, and apply it in the setting of multi-user operating systems, in which a malicious user eavesdrops on other users' keystrokes. Timing-based attacks appear also in various other settings, including: compromising the anonymity of users in networks~\cite{ParvAnonymous,ParvAnonymityChaum}, information leakage in the context of shared schedulers~\cite{SharedSchedulers} and in the context of on-chip networks~\cite{TimingSuh}. 

In this paper, we consider a stylized model of such information leakage problems, and call it the \emph{information blurring system}. The setup, shown in Figure~\ref{figurenokey}, consists of the following. 
\begin{figure}[h]
\centering
\includegraphics[scale=0.8]{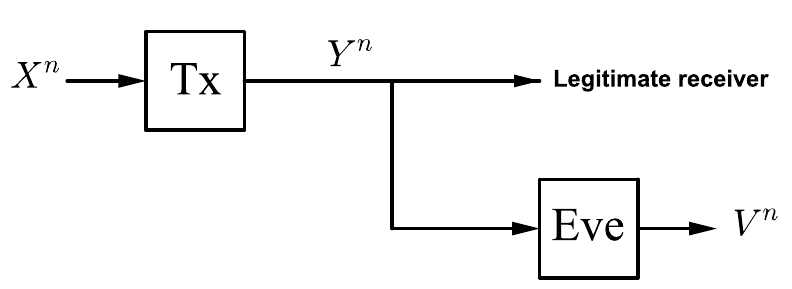}
\caption{Information blurring system: both the legitimate receiver and the eavesdropper are allowed a certain distortion level.} \label{figurenokey}
\end{figure}
A transmitter observes a sequence $X^n$, which corresponds roughly to the original timing vector, and maps it to a sequence $Y^n$ that is observed by both the legitimate receiver and an eavesdropper. The mapping must almost surely satisfy a distortion constraint, which corresponds to some quality constraints imposed by the network (e.g., delay constraints). We do not require the mapping to be causal as the intent of this work is to provide fundamental limits for a simplified version of the information leakage problem. 
In broad terms, the transmitter wants to \emph{blur} the information in $X^n$ (hence the name), so that it is no longer useful for the eavesdropper. For example, one approach is to artificially add noise to the input sequence. In that sense, the problem is related to methods for ensuring \emph{differential privacy}, in which a curator wants to publicly release statistical information about a given population without compromising the privacy of its individuals~\cite{DiffPrivacyMech,DiffPrivacySurvey}. 

Upon observing the output $Y^n$, the eavesdropper, who knows the source statistics and the transmitter's encoding function, tries to estimate $X^n$. We introduce a distortion function and consider the eavesdropper's estimate to be successful if the distortion it incurs is below a given level. Hence, we measure the \emph{secrecy} guaranteed by a given scheme via the probability that the eavesdropper makes a successful guess. The primary user (i.e., the transmitter legitimate-receiver pair) aims then to minimize that probability. Since computing the exact probability is quite difficult, this paper will be mainly concerned with asymptotic analysis: we will derive the rate of decay (i.e., the exponent) of the probability of a successful guess.  Other metrics for quantifying secrecy exist in the literature; we discuss the motivations and the shortcomings of the commonly used ones in Section~\ref{SecrecyMetric}.  
 
For a discrete memoryless source (DMS), we provide a single-letter characterization of the optimal exponent (cf. Theorem~\ref{MainThm}). We show that the problem is related to source coding with side information. Essentially, the eavesdropper first attempts to guess the joint type of $X^n$ and $Y^n$. S/he, then, ``pretends'' that $Y^n$ is received through a memoryless channel the probability law of which is the conditional probability $P(Y|X)$ induced by the joint type. The problem can be viewed at this point as compression with side information, so the eavesdropper picks a codeword from an optimal rate-distortion code. The primary user's objective, therefore, is to supply the ``worst'' side information. Moreover, we demonstrate asymptotically-optimal universal schemes for both the primary user and the eavesdropper. The schemes are universal in the sense that they do not depend on the source statistics. In particular, the transmitter operates on a type-by-type basis, and associates with each type a rate-distortion code, the construction of which is based on the conditional probability law that provides the ``worst'' side information given that type.

Next, we extend the study to the setup in which the transmitter is subject to a rate constraint, and the transmitter and the legitimate receiver have access to a common source of randomness, called the \emph{key}. The eavesdropper has full knowledge of the encryption system, except for the realization of the key and the realization of $X^n$. The setup is shown in Figure~\ref{figureShannonCipher}, and the special case in which the legitimate receiver and the eavesdropper must reconstruct the source exactly is known as the Shannon cipher system~\cite{ShannonSecrecy}.
\begin{figure}[htp]
\centering
\includegraphics[scale=0.8]{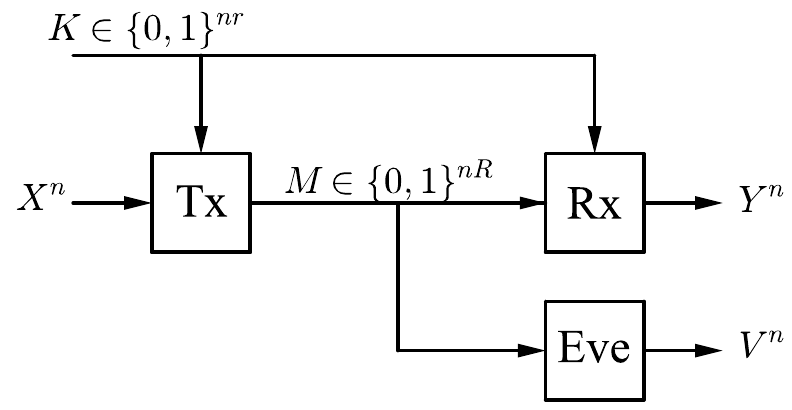}
\caption{The Shannon cipher system with lossy communication: the transmitter and the legitimate receiver have access to a common key $K$, which consists of $nr$ purely random bits, where $r$ is called the key rate. The transmitter encodes $X^n$ using $K$, and sends a message $M$ through a  noiseless public channel of rate $R$. Both the legitimate receiver and the eavesdropper are allowed a certain level of distortion. The legitimate receiver generates the reconstruction $Y^n$ based on $M$ and $K$, whereas the eavesdropper has access to $M$ only to produce an estimate $V^n$.}\label{figureShannonCipher}
\end{figure}
Since the transmitter is subject to a rate constraint, we allow the primary user to violate the distortion constraint, but restrict the probability of such event to be exponentially decaying. We again derive a single-letter characterization of the optimal exponent (cf. Theorem~\ref{Thmwithkey}), and demonstrate asymptotically-optimal strategies for both the primary user and the eavesdropper.
In particular, similarly to the previous setting, the transmitter operates on a type-by-type basis and associates with each type a rate-distortion code, the construction of which is based on the conditional probability law that provides the worst side information and satisfies the rate constraint (however, types with low enough probability are discarded, by associating a dummy messsage to all the source sequences belonging to such types). To make use of the shared key, we (randomly) generate many instances of such codes, and use the secret key to randomize the choice of the code selected for encoding $X^n$.  We also investigate conditions under which the resulting codes are  optimal rate-distortion codes.
As for the eavesdropper, we show that one of the following two schemes is optimal. The first consists of generating a blind guess, i.e., completely ignoring the public message. The second consists of guessing the value of the key to reproduce the reconstruction at the legitimate receiver, and then applying the strategy developed in the first part of the paper.
%For the primary user, we analyze the performance of the following scheme: for each type class, an optimal rate-distortion code is generated; each sequence $x^n$ is mapped to the proper reconstruction in the code corresponding to its type; and the common key is used to randomize the choice of the reconstruction within that code. For the eavesdropper, two schemes are suggested. The first consists of generating a blind guess, i.e., completely ignoring the public message. The second consists of guessing the value of the key to reproduce the reconstruction at the legitimate receiver, and then applying the strategy developed in the first part of the paper. The suggested schemes do not depend on the source statistics, and the implied upper and lower bounds admit a single-letter characterization. Moreover, we show that they match under a given condition, which is satisfied in certain cases of interest. For instance, they match if the eavesdropper must reconstruct the source sequence exactly, or if the source is binary and the distortion functions are the Hamming distance with $D_e \leq D$, where $D_e$ and $D$ are the allowed distortion levels at the eavesdropper and the legitimate receiver, respectively. They also match if the key rate is sufficiently high.

We note that Theorem~\ref{Thmwithkey} subsumes Theorem~\ref{MainThm} by setting the key rate to be zero, and the channel rate to be high enough. We nevertheless present them separately for two reasons. We believe the information blurring system to be of independent interest, as it corresponds to problems different from the Shannon cipher system (e.g., the SSH timing attack). As such, Theorem~\ref{MainThm} can serve as a baseline for future refinements of this model (say, by requiring the encoding to be causal). Moreover, it significantly simplifies the exposition of the results, by first revealing the connection to source coding with side information and then introducing the key and the rate constraint.

Finally, it should be noted that Weinberger and Merhav studied the Shannon cipher system with lossy communication~\cite{merhav2015probmetric,merhavsecurelossyTrans}  (i.e, the setup of the second part of this paper), and independently suggested the same secrecy metric  we proposed. Furthermore, they allowed a variable key rate.
In their initial work~\cite{merhav2015probmetric}, they derived the optimal exponent under the assumption that the distortion constraint of the eavesdropper is more \emph{lenient} than that of the legitimate receiver (which makes the \emph{no-key} problem degenerate). Our initial work (an earlier submission of the current paper) characterized the optimal exponent only under certain conditions (including the no-key case), which are not satisfied in the setting of~\cite{merhav2015probmetric}, and provided general upper and lower bounds. As such, those results were not comparable with that of~\cite{merhav2015probmetric}. Weinberger and Merhav later~\cite{merhavsecurelossyTrans} generalized their result to characterize the exponent in general, as is done here. However, the suggested scheme herein and its subsequent analysis are significantly simpler. In particular, our scheme uses a traditional random coding construction followed by a separate key-based randomization.

\section{Secrecy Metric} \label{SecrecyMetric}

The information-theoretic study of secrecy systems was initiated by Shannon in~\cite{ShannonSecrecy}. Shannon derived the following negative result: ensuring perfect secrecy, i.e., making the source sequence $X^n$ and the public message $M$ (cf. Figure~\ref{figureShannonCipher}) statistically independent, requires that the key rate be at least as large as the message rate.

As opposed to perfect secrecy, the notion of ``partial'' secrecy is more difficult to quantify. However, the impracticality of ensuring perfect secrecy, as implied by Shannon's result, means that developing such a notion is important from a practical point of view as well as a theoretical one. Shannon used \emph{equivocation} --- the conditional entropy of the source sequence given the public message $H(X^n|M)$ ---  as a ``theoretical secrecy index''. A main motivation for equivocation was the similarity between the deciphering problem for the eavesdropper in the secrecy setting and the decoding problem for the receiver in the standard noisy communication setting~\cite{ShannonSecrecy}. Equivocation has subsequently been used as a secrecy metric in several works~\cite{WireTapChannel,SecrecyFading,relayeavesdropper,relayhelpereve,csiszar1978broadcast,merhav2008shannon,erkippoor2008lossless}. However, its use is not well motivated operationally. It only provides a lower bound on the exponent of the list size that the eavesdropper must generate to reliably include the source sequence. Moreover, Massey showed in~\cite{GuessandEntropy} that the expected number of guesses that need to be made to correctly guess a discrete random variable $X$ may be arbitrarily large for arbitrarily small $H(X)$.   

Merhav and Arikan~\cite{MerhavArikanShannonCipher} proposed a more direct approach: they consider an i.i.d. source and they measure secrecy by the expected number of guesses that the eavesdropper needs to make before finding the correct source sequence, which they denote by $\E[G(X^n|M)]$, where $G(.|m)$ is a ``guessing'' function defined for each possible public message $m$. This is intended to capture the scenario in which the eavesdropper has a testing mechanism to check whether or not his/her guess is correct. Such mechanism exists, for example, if the source message is a password to a computer account. When the source is discrete and memoryless, and the transmitter and the legitimate receiver have access to $nr$ purely random common bits (where $r$ is called the key rate), the optimal exponent of $\E[G(X^n|M)]$ is found to be~\cite[Theorem 1]{MerhavArikanShannonCipher}:
\begin{align} 
E(P,r) & \triangleq \lim_{n \rightarrow \infty} \frac{1}{n} \log \E[G(X^n|M)] \notag \\
& = \max_Q \left\lbrace \min\{H(Q),r\} -D(Q||P) \right\rbrace, \label{MerhavArikanExponent}
\end{align}
where $P$ is the source distribution and $D(\cdot||\cdot)$ is the Kullback-Leibler (KL) divergence. Two issues arise with this metric. First, even if a testing mechanism exists, any practical system would only allow a small number of incorrect inputs. Thus, it is not clear how to interpret an exponentially large number of guesses. Second, and more importantly, it turns out that even highly-insecure systems can appear to be secure under this metric.
Indeed, by modifying the asymptotically-optimal scheme proposed in~\cite{MerhavArikanShannonCipher}, we can construct a scheme for the primary user that allows the eavesdropper to find the source sequence correctly with high probability by the first guess, and yet achieves the optimal exponent in~\eqref{MerhavArikanExponent}. The scheme proposed in~\cite{MerhavArikanShannonCipher} operates on the source sequences on a type-by-type basis, and it yields:
\begin{equation}
\E \left[G(X^n|M) \big | X^n \in T_Q \right] \geq 2^{n \min\{r,H(Q)\}-o(n)},
\end{equation}
where $o(n)/n \rightarrow 0$ as $n \rightarrow \infty$, and $T_Q$ is the type class of a given type $Q$, i.e., the set of sequences with empirical distribution $Q$. Averaging over the probabilities of $\{T_Q\}$ yields the exponent in~\eqref{MerhavArikanExponent} (as a lower bound). However, this means that it is enough to apply the proposed scheme to the type class $T_Q$ that achieves the maximum of $ \left[ \min\{H(Q),r\} -D(Q||P) \right]$, whereas sequences belonging to other type classes can be sent with no encoding whatsoever with no effect on the exponent. Therefore, only a set with vanishing probability is encoded, whereas sequences outside that set are immediately known by the eavesdropper\footnote{Merhav and Arikan actually characterize, for any $\rho > 0$, the exponent of $\E[G^\rho (X^n|M)]$. This more general result can still yield large exponents for systems that are highly insecure, although one could potentially address this issue by requiring schemes that yield large exponents simultaneously over a range of $\rho$ values.}. 

A different approach, based on rate-distortion theory, was adopted by Yamamoto in~\cite{yamamoto1997shannoncipher}. A distortion function is introduced and the secrecy of a given scheme is measured by the minimum attainable expected distortion at the eavesdropper. Also, a certain level of distortion, possibly corresponding to a different distortion function, is allowed at the legitimate receiver. An earlier work by Yamamoto~\cite{yamamoto1988secdec} considered the special case where no key is available, under the same secrecy metric. 
A standard example, discussed and generalized in~\cite{cuff2013secrecycausal}, shows why expected distortion is  inadequate: Suppose $X^n$ is a sequence of independent and identically distributed bits with $X_i \sim$ Ber($1/2$), the transmitter and the legitimate receiver have access to one common bit $K \sim$ Ber($1/2$), and the distortion function is the Hamming distance. The transmitter then sends the sequence $X^n$ as is if $K=0$, and flips all its bits if $K=1$. The induced expected distortion at the eavesdropper is then equal to $1/2$, which is also the maximum expected distortion that the eavesdropper can possibly incur, since it is achievable even if the public message is not observed. However, this ``optimal'' scheme in fact reveals a lot about the true source sequence; namely, it is one of only two possible candidates. 

To overcome this limitation of expected distortion, Schieler and Cuff~\cite{cuffisit2014henchman,cuff2014henchman} allow the eavesdropper to generate an exponentially-sized list of estimates and propose the expected minimum distortion over the list as a secrecy metric. It is not clear, however, how to operationally interpret a list of \emph{exponential} size. It is shown that this setting is equivalent to the following: there exists a ``henchman'' that has access to the source sequence $X^n$ and public message $M$, and can transmit $nR_L$ bits to the eavesdropper who measures secrecy by the minimum expected distortion.
However, this metric leads to a degenerate trade-off between the key rate $r$, the allowed list exponent (henchman rate) $R_L$, and the expected minimum distortion in the list $D_e$. For example, if the legitimate receiver must reconstruct $X^n$ losslessly, one of two cases occurs (see \cite[Theorem 1]{cuff2014henchman}): If $r > R_L$, $D_e$ is given by the rate-distortion function at $R_L$. If $r \leq R_L$, $D_e = 0$ since the eavesdropper can trivially find the exact sequence by listing all the possible keys. This fails to capture that, even when $R_L < r$, the eavesdropper can still list $2^{nR_L}$ possible keys and thus recover exactly the correct sequence with probability at least $2^{n(R_L-r)}$. As $R_L$ approaches $r$, this probability can be made to decay arbitrarily slowly. It is worth noting that Schieler and Cuff also consider a \emph{causal disclosure} setting~\cite{cuff2013secrecycausal}, in which the eavesdropper observes, at time $i$, the public message $M$ and $X^{i-1}$. Although this is more robust than expected distortion, it captures only a limited range of practical scenarios.
%either the public message $M$ is made completely useless to the eavesdropper when $r > R_L$, and $D_e$ is then given by the distortion-rate function at $R_L$; or, when $r \leq R_L$, the eavesdropper can trivially find the exact sequence by listing all the possible keys. Even when $R_L < r$, the eavesdropper can still list $2^{nR_L}$ possible keys and thus recover exactly the correct sequence with probability at least $2^{n(R_L-r)}$. As $R_L$ approaches $r$, this probability can be made to decay arbitrarily slowly.

In this paper, we take a different approach. In many applications, the eavesdropper has no way to verify if his/her estimate is correct. This is particularly true in our main case of interest, i.e., timing of events. Moreover, as mentioned before, most practical systems allow a small number of incorrect guesses even if a testing mechanism exists. Therefore, we allow the eavesdropper to make one guess only. Secrecy is measured then by the probability that the guess is successful, i.e., the distortion incurred is below a given level. For the purposes of the asymptotic analysis in this paper, we will study only the exponent of the probability of a successful guess. A special case of such analysis was considered by Merhav in~\cite{merhav2003perfectsecrecy}. In particular,~\cite{merhav2003perfectsecrecy} is concerned with necessary and sufficient conditions for achieving the perfect secrecy exponent, which is the exponent attained by the eavesdropper in the absence of any observation. It is also restricted to the case in which both the legitimate receiver and the eavesdropper must reconstruct the source sequence exactly. Finally, a relevant earlier work by Arikan and Merhav~\cite{arikanguessingdistorion} considers the problem of blindly guessing a random variable up to a distortion level and characterizes the least achievable exponential growth of the expected number of guesses.

\section{Information Blurring System} \label{nokey}

We consider the following secrecy system. Let $\mathcal{X}$, $\mathcal{Y}$, and $\mathcal{V}$ be the alphabets associated with the transmitter, the legitimate receiver, and the eavesdropper, respectively. The transmitter wants to provide the legitimate receiver with a quantized version of an $n$-length message $X^n=(X_1, X_2, \cdots, X_n) $. It thus generates a vector $Y^n$ through a (possibly randomized) function $f:~ Y^n=f(X^n)$. For a given distortion function $d: \mathcal{X} \times \mathcal{Y} \rightarrow \mathbb{R}_+$, the quantization is required to satisfy a constraint of the form $d(X^n,Y^n)= \frac{1}{n} \sum_{i=1}^n d(X_i,Y_i) \leq D$ for a given distortion level $D$. The restriction is imposed on each realization of $(X^n,Y^n)$. An eavesdropper, with an associated distortion function $d_e: \mathcal{X} \times \mathcal{V} \rightarrow \mathbb{R}_+$, also observes $Y^n$ and generates a guess $V^n=g(Y^n)$, aiming to have $d_e(X^n,V^n) = \frac{1}{n} \sum_{i=1}^n d_e(X_i,V_i) \leq D_e$ for a given distortion level $D_e$.

It is assumed that the eavesdropper knows the source statistics and the primary user's encoding function $f$. The secrecy metric we adopt is the probability that the eavesdropper makes a successful guess, i.e., $\Pr( d_e(X^n,V^n) \leq  D_e)$. 
%the eavesdropper maximizes the probability of a successful guess  by implementing the MAP rule, which we denote by\footnote{The MAP rule depends on $f$ and thus should be denoted by $g_{o,f}$. Since this is obvious, we drop the subscript $f$ for notational convenience.} $g_{o}$ (where ``o'' stands for optimal). 
The primary user's objective is to minimize this probability. So, the problem can be written as:
\begin{equation*}
\min_{f_n}  \max_{g_n} \Pr \left( d_e\Big(X^n,g_n(f_n(X^n))\Big) \leq D_e \right).
\end{equation*}

We characterize the highest achievable exponent of the probability of a successful guess under the following assumptions:
\begin{enumerate}
\item[(A1)] The alphabets $\mathcal{X}$, $\mathcal{Y}$ and $\mathcal{V}$ are finite.
\item[(A2)] The source is memoryless, and without loss of generality, its distribution has full support.
\item[(A3)] The distortion functions $d$ and $d_e$ are bounded, i.e., there exists $D_{\max}$ and $D_{e,\max}$ such that, for all $x \in \mathcal{X}$, $y \in \mathcal{Y}$, and $v \in \mathcal{V}$, $d(x,y) \leq D_{\max}$ and $d_e(x,v) \leq D_{e,\max}$. Moreover, $ D \geq D_{\min}$, where $D_{\min} = \max_{x \in \mathcal{X}} \min_{y \in \mathcal{Y}} d(x,y)$. Similarly,  $ D_e \geq D_{e,\min}$, where $D_{e,\min} = \max_{x \in \mathcal{X}} \min_{v \in \mathcal{V}} d_e(x,v)$.
\end{enumerate}

We denote the optimal exponent by $E(P,D,D_e)$, where $P$ is the source distribution, i.e.,
\begin{align} 
& E(P,D,D_e) = \notag \\
&  \lim_{n \rightarrow \infty} \max_{\{f_n\}} \min_{\{g_n\}} -\frac{1}{n}\log \Pr \left( d_e \Big(X^n,g_n(f_n(X^n)) \Big) \leq D_e \right),\label{defE}
\end{align} where $\{f_n\}$ is restricted to the class of functions ensuring the feasibility of the primary user's problem.
The existence of the limit will be seen later.

 We will show that the problem is related to source coding with side information, where $Y^n$ acts as side information for the eavesdropper. Therefore, the primary user's job is to provide the ``worst'' side information subject to a distortion constraint of his/her own. 
To this end, we denote the \emph{conditional} rate-distortion function  as:
\begin{equation} \label{Rsideinfo}
R(P_{XY},D_e) = \min_{ \substack{P_{V|X,Y}:\\ ~~\E[d_e(X,V)] \leq D_e}} I(X;V |Y),
\end{equation}
and define the quantity $R(P_X,D,D_e)$ as:
\begin{equation} \label{Rmaxmin}
R(P_X,D,D_e) =  \max_{ \substack{ P_{Y|X}: \\ \E[d(X,Y)] \leq D}} R(P_{XY},D_e).
\end{equation}
Roughly speaking, when the joint type of $X^n$ and $Y^n$ is $P_{XY}$, the eavesdropper can restrict the guessing space to $2^{nR(P_{XY},D_e)}$ reconstruction sequences, knowing that at least one of them must satisfy the distortion constraint. The maximization in~\eqref{Rmaxmin} corresponds to the primary user's goal of maximizing that quantity.   \\

We prove the following properties of $R(P_{XY},D_e)$ and $R(P_X,D,D_e)$ in Appendix~\ref{continuityproof}.
\begin{Proposition} \label{continuity}
In the following statements, the domains of $D$ and $D_e$ are $[D_{\min},+\infty)$ and $[D_{e,\min},+\infty)$, respectively.
\begin{enumerate}
\item[(P1)] For fixed $P_{XY}$, $R(P_{XY},D_e)$ is a finite-valued, non-increasing convex function of $D_e$. Furthermore, $R(P_{XY},D_e)$ is a uniformly continuous function of the pair $(P_{XY},D_e)$.
\item[(P2)] For fixed $P_X$, $R(P_X,D,D_e)$ is a finite-valued function of $(D,D_e)$. Moreover, for fixed $D_e$, $R(P_X,D,D_e)$ is a uniformly continuous function of the pair $(P_X,D)$.
\item[(P3)] $R_e(P_X,D_e) - R(P_X,D) \leq R(P_X,D,D_e) \leq R_e(P_X,D_e)$, where $R(P_X,D)$ and $R_e(P_X,D_e)$ are the rate-distortion functions corresponding to the distortion constraints $d$ and $d_e$, respectively.
\end{enumerate}
\end{Proposition}
%\begin{Remark}
%Property (P1) justifies the use of ``max'' in~\eqref{Rmaxmin} (since a continuous function on a compact set achieves its maximum).
%\end{Remark}
\bigskip
Our main result is the characterization of the optimal exponent as follows:
\begin{Theorem} \label{MainThm}
Under assumptions (A1)-(A3), for any DMS $P$, and distortion functions $d$ and $d_e$ with associated distortion levels $D \geq D_{\min}$ and $D_e \geq D_{e,\min}$, corresponding respectively to the primary user and the eavesdropper:
\begin{equation} \label{maineq}
E(P,D,D_e) = \min_{Q} D(Q||P) + R(Q,D,D_e), 
\end{equation}
where $Q$ ranges over all probability distributions on the source alphabet, and $R(Q,D,D_e)$ is as defined in \eqref{Rmaxmin}.
\end{Theorem}
%\bigskip
\begin{Remark}
We do not require any $\epsilon$-backoff for $D$ or $D_e$ to characterize the associated exponent.
\end{Remark}
\bigskip

An interesting feature of Theorem~\ref{MainThm} is the emergence of mutual information as part of the solution in~\eqref{maineq}, even though the setup does not include any rate constraints. Moreover, an interesting contrast can be seen between the expression in~\eqref{MerhavArikanExponent} for the expected number of guesses metric and the expression in~\eqref{maineq} for our metric. Indeed, the former  evaluates the performance of a given scheme asymptotically by a weighted \emph{best-case} scenario, whereas the latter evaluates it by a weighted \emph{worst-case} scenario. 

As an application of the theorem, we compute the \emph{perfect secrecy} exponent, which we define as the best achievable exponent when the primary user is not subject to any constraint and denote it by $E_0(P,D_e)$. To this end, we introduce a trivial distortion function: $d(x,y)=0$, for all $x \in \mathcal{X}$ and $y \in \mathcal{Y}$. Then, $R(Q,D)=0$, for all $Q$ and all $D \geq 0$. It then follows from (P3) of Proposition~\ref{continuity} that $R(Q,D,D_e)=R_e(Q,D_e)$ for all $Q$. Therefore,
\begin{equation}
\label{eqperfectsecrecyexp}
E_0(P,D_e) = \min_Q D(Q||P)+R_e(Q,D_e).
\end{equation} 
%We claim that this is the optimal exponent achievable by the eavesdropper when he/she needs to guess in the absence of any observation. This is easy to see if the eavesdropper must reconstruct the source exactly. In that case, $R_e(Q,0)=H(Q)$ for all $Q$, and $E_0(P,0)$ is simply given by the exponent of the most likely sequence according to the \emph{prior} distribution of the source. To prove the claim for any distortion constraint, let $E_\varnothing(P,D_e)$ be the best achievable exponent for the eavesdropper in the absence of any observation. Then, $ E_0(P,D_e) \leq E_\varnothing(P,D_e) $, since the eavesdropper can always choose to ignore the output $Y^n$. On the other hand, a valid scheme for the primary user, when the trivial constraint is imposed, is to deterministically  send a fixed sequence $y^n$ for some $y \in \mathcal{Y}$. Let $E_c(P,D_e)$ be the exponent associated with that scheme. Then, $E_c(P,D_e) \leq E_0(P,D_e)$. Finally, note that $P(X^n| y^n) = P(X^n)$. Therefore, $E_c(P,D_e) = E_\varnothing (P,D_e)$, yielding $E_0(P,D_e) = E_\varnothing (P,D_e)$. \\

The next two subsections are devoted to proving Theorem~\ref{MainThm}. We first propose a scheme for the primary user and show that the induced exponent is lower-bounded by the right-hand side of~\eqref{maineq}. From the eavesdropper's point of view, this is a converse result. Similarly, we propose a scheme for the eavesdropper and show that the induced exponent is upper-bounded by the right-hand side of~\eqref{maineq}, which establishes the desired result.

We set some notation for the remainder of the paper. In the following, $\mathcal{Z}$ is an arbitrary discrete set, and $Z$ is a random variable over $\mathcal{Z}$.
\begin{enumerate}
\item[-] The set of probability distributions over $\mathcal{Z}$ is denoted by $\mathcal{P}_\mathcal{Z}$. 
\item[-] For a sequence $z^n \in \mathcal{Z}^n$, $Q_{z^n}$ is the empirical PMF of $z^n$, also referred to as its type.
\item[-] $\mathcal{Q}_\mathcal{Z}^n$ is the set of types in $\mathcal{Z}^n$, i.e., the set of rational PMF's with denominator $n$.
\item[-] For $Q_Z \in \mathcal{Q}_\mathcal{Z}^n$, the type class of $Q_Z$ is $T_{Q_Z} \triangleq \{z^n \in \mathcal{Z}^n: Q_{z^n}=Q_Z \}$.
\item[-] $\E_Q[\cdot]$, $H_Q(\cdot)$, and $I_Q(\cdot;\cdot)$ denote respectively expectation, entropy, and mutual information taken with respect to distribution $Q$.
\item[-] All logarithms and exponentials are taken to the base 2.
\end{enumerate}

\subsection{Achievability for the Primary User (Eavesdropper's Converse Result)} \label{primaryachieve}

Let \begin{align} 
& E^-(P,D,D_e) = \notag \\
& \liminf_{n \rightarrow \infty} \max_{\{f_n\}} \min_{\{g_n\}} -\frac{1}{n}\log \Pr \left( d_e \Big(X^n,g_n(f_n(X^n))\Big) \leq  D_e \right).\label{defE-}
\end{align}  
We will show that $E^-(P,D,D_e) \geq \min_{Q} D(Q||P) + R(Q,D,D_e)$. 

The primary user will operate on the source sequences on a type-by-type basis. For each type $Q_X \in \mathcal{Q}_\mathcal{X}^n$, we create a rate distortion code $\mathcal{C}_{Q_X}$ to cover each sequence in $T_{Q_X}$ as follows. We associate with $Q_X$ a joint type $Q_{XY}$ from $\mathcal{Q}_{\mathcal{XY}}^n(Q_X,D)$:\footnote{Assumption (A3) guarantees that $\mathcal{Q}_{\mathcal{XY}}^n(Q_X,D)$ is nonempty for any $Q_X$.}
\begin{align}
& \mathcal{Q}_{\mathcal{XY}}^n(Q_X,D) = \notag \\
& ~~\{ P_{XY} \in \mathcal{Q}_{\mathcal{XY}}^n: P_X=Q_X, ~ \E_{P_{XY}} [d(X,Y)]  \leq D\}.
\label{defqqxd}
\end{align}  
The code is then constructed from $T_{Q_Y}$ as given by the following lemma, which bounds the size of the code. 
\bigskip
\begin{Lemma} \label{lemmacode}
Given $\epsilon > 0$, there exists $n_0\left(\epsilon,|\mathcal{X}|,|\mathcal{Y}| \right)$ such that for any $n \geq n_0$, for each joint type $Q_{XY} \in \mathcal{Q}_{\mathcal{XY}}^n$, there exists a code $(y_1^n,y_2^n,\cdots,y_N^n)$ such that $N \leq 2^{n( I_{Q_{XY}} (X;Y)+\epsilon)}$, and for all $x^n \in T_{Q_X},$ there exists $i$ satisfying $(x^n,y_i^n) \in T_{Q_{XY}}$.
\end{Lemma}
The proof is a refinement of the covering lemma~\cite[Lemma 2.4.1]{korner}. We later prove a stronger result, Lemma~\ref{lemmaachievability}, in Appendix~\ref{mainlemmaproof}.
%The proof of the lemma is given in Appendix~\ref{lemmacodeproof}.
\bigskip
\begin{Remark}
One might be tempted to use an optimal rate-distortion code for each type $Q_X$, presuming that this choice is best at preserving secrecy since it achieves optimal compression, i.e., it only sends the necessary information. However, the problem is more subtle since the ``redundancy of information'' depends on the eavesdropper's distortion constraint $d_e$. The optimal choice of $Q_{XY}$ will be revealed when analyzing the eavesdropper's optimal strategy. 
\end{Remark}
\bigskip

Now, fix $\epsilon >0$ and let $n$ be at least as large as $n_0$ in Lemma~\ref{lemmacode}. We will denote by $\mathcal{C}^n_{Q_X}$ the rate distortion code associated with type $Q_X$. Thus, the function $f$ of the primary user is as follows: each sequence $x^n$ is mapped to a sequence $y^n \in \mathcal{C}^n_{Q_{x^n}}$ satisfying $Q_{x^ny^n}=Q_{XY}$ (where $Q_{XY}$ is associated with $Q_{x^n}$) and subsequently $d(x^n,y^n)\leq D$.\\

To determine the eavesdropper's optimal guess, we define $B_{D_e} (v^n)= \{ x^n \in \mathcal{X}^n: ~d_e(x^n,v^n) \leq  D_e \}.$ Then, for each observed $y^n$, the optimal rule is given by
\begin{equation*}
g(y^n) = \argmax_{v^n \in \mathcal{V}^n}  \sum_{x^n \in B_{D_e} (v^n) } p(x^n | y^n).
\end{equation*}
This can be understood as the MAP rule, and we denote in the remainder by\footnote{The MAP rule depends on $f$ and thus should be denoted by $g_{o,f}$. Since this is obvious, we drop the subscript $f$ for notational convenience.} $g_o$ (where ``o'' stands for optimal). To upper-bound the probability of a correct guess, we consider a \emph{genie-aided} rule that is aware of the type of the transmitted source sequence. That is, the genie-aided MAP rule yields 
\begin{equation*} 
g_o(y^n,Q_X)= \argmax_{v^n \in \mathcal{V}^n}   \sum_{x^n \in B_{D_e} (v^n) \cap Q_X }  p(x^n | y^n, X^n \in T_{Q_X}).  
\end{equation*}
\begin{Remark}
One should not expect the upper bound to be loose since there are only polynomially many types in $n$, so that the exponent is not affected.  
\end{Remark}
For a given $y^n$, let $f^{-1}_{Q_X}(y^n)=\{x^n \in T_{Q_X}: f(x^n)=y^n\}$ be the set of sequences in $T_{Q_X}$ that are mapped to it. Then, the observation of $y^n$ implies that $X^n \in f_{Q_X}^{-1}(y^n)$, and the genie-aided MAP rules makes a successful guess if $X^n \in B_{D_e} \left(g_o(y^n,Q_X) \right)$. Therefore, we will derive an upper bound on the maximum possible size of the intersection of these two sets. First, note that, $x^n \in T_{Q_X}$ and $f(x^n)=y^n$ implies that $Q_{x^ny^n}=Q_{XY}$, where $Q_{XY}$ is the joint type associated with $Q_X$. So $f^{-1}_{Q_X} (y^n) \subseteq T_{Q_{X|Y}}(y^n) \triangleq \{x^n \in T_{Q_X}: (x^n,y^n) \in T_{Q_{XY}} \}$. Now, consider any $v^n \in \mathcal{V}^n$,
\begin{align} \label{eavesbound1}
& \left|B_{D_e}(v^n) \bigcap f_{Q_X}^{-1}(y^n) \right| \notag \\
& 
\leq  \left|B_{D_e}(v^n) \bigcap T_{Q_{X|Y}}(y^n) \right| \notag \\ 
 & \stackrel{\text{(a)}} =  \sum_{\substack{P_{XYV} \in \mathcal{Q}_{\mathcal{XYV}}^n: \\ P_{XY} = Q_{XY} \\ \E_{P_{XYV}} [d_e(X,V)] \leq D_e \\ P_{YV}=Q_{y^n v^n} } } \sum_{ \substack{x^n: \\ (x^n,y^n,v^n) \in T_{P_{XYV}} }} 1 \notag \\
& \stackrel{\text{(b)}} = \sum_{ \substack{P_{XYV} \in \mathcal{Q}^n(Q_{XY},D_e): \\ P_{YV}=Q_{y^n v^n}} } |T_{P_{X|V,Y}}(v^n,y^n) |  \notag \\
& \leq (n+1)^{|\mathcal{X}||\mathcal{Y}| |\mathcal{V}|} \max_{\substack{P_{XYV} \in \mathcal{Q}^n(Q_{XY},D_e): \\ P_{YV}=Q_{y^n v^n}}  } |T_{P_{X|V,Y}}(v^n,y^n) | \notag \\
& \stackrel{\text{(b)}}\leq (n+1)^{|\mathcal{X}||\mathcal{Y}| |\mathcal{V}|} \max_{\substack{P_{XYV} \in \mathcal{Q}^n(Q_{XY},D_e): \\ P_{YV}=Q_{y^n v^n}}  } 2^{n H_{P_{XYV}} ( X|V,Y) },
\end{align}
where  
\begin{enumerate}
\item[(a)] follows from the fact that $ (x^n,y^n,v^n) \in T_{P_{XYV}} \Rightarrow P_{YV}=Q_{y^n v^n}$.
\item[(b)] follows from the definition of $\mathcal{Q}^n(Q_{XY},D_e)$ as:
%\begin{align} 
%& \mathcal{Q}^n(Q_{XY},D_e) \notag \\
%&  ~~ = \{ P_{XYV} \in \mathcal{Q}_{\mathcal{XYV}}^n: P_{XY}=Q_{XY}, \notag \\
%&  ~~ \quad \qquad \qquad \qquad \qquad \E_{P_{XYV}} [d_e(X,V)] \leq D_e \}. \label{QXYDE}
%\end{align} 
\begin{align} 
 \mathcal{Q}^n(Q_{XY},D_e) = \{ & P_{XYV} \in \mathcal{Q}_{\mathcal{XYV}}^n: P_{XY}=Q_{XY}, \notag \\
&  \E_{P_{XYV}} [d_e(X,V)] \leq D_e \}. \label{QXYDE}
\end{align} 
\item[(b)] follows from Lemma 1.2.5 in~\cite{korner}.
\end{enumerate}
Therefore, for large enough $n$, we get
\begin{align}
&\max_{v^n \in \mathcal{V}^n}  \left|B_{D_e}(v^n) \bigcap f_{Q_X}^{-1}(y^n) \right| \notag \\
& \leq  \max_{P_{XYV} \in \mathcal{Q}^n(Q_{XY},D_e)} 2^{ n \left(H_{P_{XYV}} ( X|V,Y)+\epsilon \right) }.\label{eavesbound}
\end{align}
Let $P^\star_n(Q_{XY})$ be the joint type achieving the max in \eqref{eavesbound}, where the dependence on $D_e$ is suppressed since it is fixed throughout the analysis. We can now upper-bound the probability that the eavesdropper makes a successful guess as follows:
\begin{align} \label{mainchain1}
& \Pr \Big(d_e \Big(X^n, g_{o}(f(X^n)) \Big) \leq  D_e \Big) \notag \\
 & \leq \Pr \Big(d_e \Big(X^n,g_o \left(f(X^n),Q_{X^n} \right) \Big) \leq  D_e \Big)
\notag \\
& = \sum_{x^n \in \mathcal{X}^n}  P(x^n) \mathbf{1} \left\lbrace x^n \in B_{D_e}\left(g_o \big(f(x^n),Q_{x^n} \big) \right) \right\rbrace \notag \\
& = \! \sum_{Q_X \in \mathcal{Q}_\mathcal{X}^n} \! \sum_{y^n \in \mathcal{C}^n_{Q_X} } \! \! \sum_{ \substack{x^n \in T_{Q_X} : \\ f(x^n)=y^n}} \! \! \! P(x^n) \mathbf{1} \left\lbrace x^n \! \in \! \! B_{D_e} \! \left(g_o \big(y^n,Q_{X} \big)\!\right)\! \right\rbrace \notag \\
& \stackrel{\text{(a)}} \leq  \sum_{Q_X \in \mathcal{Q}_\mathcal{X}^n}  \sum_{y^n \in \mathcal{C}^n_{Q_X} }  2^{n \left(-D(Q_X || P)-H_{Q_X}(X) \right)}\cdot \notag \\
&  \qquad \qquad \qquad 2^{n\left(H_{P^\star_n(Q_{XY}) } (X|V,Y) +\epsilon \right)} \notag \\
& \stackrel{\text{(b)}} \leq   \sum_{Q_X \in \mathcal{Q}_\mathcal{X}^n} 2^{n \left(I_{Q_{XY}}(X;Y)+\epsilon-D(Q_X || P)-H_{Q_X}(X) \right)} \cdot \notag \vspace{-3mm} \\ 
& \qquad \qquad \qquad 2^{n\left(H_{P^\star_n(Q_{XY}) } (X|V,Y) +\epsilon \right)} \notag \\
& =  \sum_{Q_X \in \mathcal{Q}_\mathcal{X}^n} 2^{n \left(-D(Q_X || P) -H_{Q_{XY}}(X|Y)+H_{P^\star_n(Q_{XY})} (X|V,Y)+2\epsilon   \right)} \notag \\
& =  \sum_{Q_X \in \mathcal{Q}_\mathcal{X}^n} 2^{-n \left(D(Q_X||P)+I_{P^\star_n(Q_{XY})}(X;V|Y) -2\epsilon \right)},
\end{align}
where
\begin{enumerate}
\item[(a)] follows from \eqref{eavesbound}.
\item[(b)] follows from Lemma~\ref{lemmacode}.
\end{enumerate}

To interpret the exponent in~\eqref{mainchain1}, note that $P^\star_n(Q_{XY})$ minimizes $I(X;V|Y)$ over $\mathcal{Q}^n(Q_{XY},D_e)$ (follows readily from~\eqref{eavesbound}). Therefore, $I_{P^\star_n(Q_{XY})}(X;V|Y)$ is roughly $R(Q_{XY},D_e)$. The eavesdropper's scheme can then be seen as picking a codeword from an optimal rate-distortion code that uses side information generated according to $Q_{Y|X}$. 

Since $Q_{XY}$ is the choice of the primary user, who is interested in maximizing the exponents in \eqref{mainchain1}, we define for each $Q_X \in \mathcal{Q}_\mathcal{X}^n$:
\begin{equation} \label{eqQstarQX}
Q^\star(Q_X) \in \argmax_{Q_{XY} \in \mathcal{Q}_{\mathcal{XY}}^n(Q_X,D)} I_{P^\star_n(Q_{XY})}(X;V|Y),
\end{equation}
where we have again suppressed the dependence on $D$ and $D_e$ in the notation.
\begin{Remark}
The maximization does not depend on the source statistics, and consequently neither does the proposed encoding function $f$.
\end{Remark}
With a slight abuse of notation, we rewrite $P^\star_n(Q^\star(Q_X))$ as $P^\star_n(Q_X)$ to get
\begin{align} 
I_{P^\star_n(Q_X)}(X;V|Y) & = \max_{ \substack{ Q_{XY} \in \\  \mathcal{Q}_{\mathcal{XY}}^n(Q_X,D)} } I_{P^\star_n(Q_{XY})}(X;V|Y) \notag \\
& = \!\! \max_{ \substack{ Q_{XY} \in \\ \mathcal{Q}_{\mathcal{XY}}^n(Q_X,D)~ } } \!\! \min_{ \substack{ Q_{XYV} \in \\ ~\mathcal{Q}^n(Q_{XY},D_e) }}  \!\!\! I_{Q_{XYV}}\!\!(X;V|Y). \label{Imax}
\end{align}
We can now rewrite \eqref{mainchain1} as
\begin{align} \label{mainchain2}
 & \Pr \Big(d_e \Big(X^n, g_{o}(f(X^n)) \Big) \leq  D_e \Big) \notag \\
 & \leq  \sum_{Q_X \in \mathcal{Q}_\mathcal{X}^n} 2^{-n \left(D(Q_X||P)+I_{P^\star_n(Q_{X})}(X;V|Y) -2\epsilon \right)} \notag \\
 & \leq (n+1)^{|\mathcal{X}|} \max_{Q_X \in \mathcal{Q}_\mathcal{X}^n} 2^{-n \left(D(Q_X||P)+I_{P^\star_n(Q_{X})}(X;V|Y) -2\epsilon \right)}.
 \end{align}
Taking the limit as $n$ goes to infinity, and noting that $\epsilon$ is arbitrary, we get
\begin{align} \
& E^-(P,D,D_e)  = \notag \\
&  \liminf_{n \rightarrow \infty} \max_{\{f_n\}} \min_{\{g_n\}} -\frac{1}{n} \log  \Pr \Big(d_e \Big(X^n, g_n(f_n(X^n)) \Big) \leq  D_e \Big) \notag \\
& \geq \min_{Q} D(Q||P)+R(Q,D,D_e), \label{acheq}
\end{align}
where the last inequality follows from the following proposition, the proof of which is given in Appendix~\ref{proplimitproof}.
\begin{Proposition} \label{proplimit}
\begin{align*}
& \lim_{n \rightarrow \infty }  \min_{Q_X \in \mathcal{Q}_\mathcal{X}^n} [D(Q_X||P)+I_{P^\star_n(Q_{X})}(X;V|Y)] \\
& = \min_{Q} D(Q||P)+R(Q,D,D_e).
\end{align*}
\end{Proposition}

\subsection{Converse for the Primary User (Eavesdropper's Achievability Result)} \label{primaryconv}

Let $E^+(P,D,D_e)=$ \begin{align} 
 ~\limsup_{n \rightarrow \infty} \max_{\{f_n\}} \min_{\{g_n\}} 
 -\frac{1}{n}\log \Pr \left( d_e \Big( \! X^n,g_n(f_n(X^n))\! \Big) \! \leq \! D_e \right) \! . \label{defE+}
\end{align}  
We will now show that $E^+(P,D,D_e) \leq \min_{Q} D(Q||P) + R(Q,D,D_e)$. This means that the eavesdropper can achieve the exponent in \eqref{maineq} for any function $f$ the primary user implements.\\

We propose a two-stage scheme for the eavesdropper. In the first stage, observing $y^n$, s/he tries to guess the joint type of $x^n$ and $y^n$ by choosing an element uniformly at random from the set $\mathcal{Q}_{\mathcal{\mathcal{XY}}}^n(Q_{y^n},D)$, where
\begin{align}
\mathcal{Q}_{\mathcal{XY}}^n(Q_{Y},D) = \{P_{XY} \in \mathcal{Q}_{\mathcal{XY}}^n: & P_Y=Q_{Y}, \notag \\
&  \E_{P_{XY}}[d(X,Y)] \leq D \}.
\end{align}
The correct joint type must fall in this set since the restriction $d(X^n,Y^n) \leq D$ is imposed on each realization of $(X^n,Y^n)$. We  denote the function corresponding to this stage by  $g_1: \mathcal{Y}^n \rightarrow \mathcal{Q}_{XY}^n$.
\begin{Remark}
We differentiate between $\mathcal{Q}_{\mathcal{XY}}^n(Q_{Y},D)$ and $\mathcal{Q}_{\mathcal{XY}}^n(Q_{X},D)$ by their first argument. A summary of different notations is given in Table~\ref{tablesets}.
\end{Remark}
\begin{table}[t]
\caption{Summary of useful notation.} \label{tablesets}
\centering
\begin{tabular}{| c | c |}
\hline
Notation & Description \\
\hline
$R(P_{XY},D_e)$ & $\min_{ {P_{V|X,Y}:\E[d_e(X,V)] \leq D_e}} ~~I(X;V |Y)$ \\
\hline
$R(P_X,D,D_e)$ & $~\max_{ { P_{Y|X}:  \E[d(X,Y)] \leq D}} ~~R(P_{XY},D_e)$ \\
\hline
$B_{D_e}(v^n)$ & $x^n \in \mathcal{X}^n: d_e(x^n,v^n) \leq D_e$ \\
\hline
$\mathcal{Q}_{\mathcal{XY}}^n(Q_X,D)$ & $ P_{XY} \in \mathcal{Q}_{\mathcal{XY}}^n:  \!\! P_X=Q_X, ~ \E_{P_{XY}} [d(X,Y)]  \leq D.$ \\
\hline
$\mathcal{Q}_{\mathcal{XY}}^n(Q_{Y},D) $ & $ P_{XY} \in \mathcal{Q}_{\mathcal{XY}}^n: \!\!  P_Y=Q_{Y},~ \E_{P_{XY}}[d(X,Y)] \leq D. $ \\
\hline
%$\mathcal{Q}^n(Q_{XY},D_e)$ & $ P_{XYV} \in \mathcal{Q}_{\mathcal{XYV}}^n: P_{XY}=Q_{XY},~ \E_{P_{XYV}} [d_e(X,V)] \leq D_e. $ \\
%\hline
$\mathcal{Q}^n(Q_{XY},D_e)$ & \begin{tabular}{c}
$ P_{XYV} \in \mathcal{Q}_{\mathcal{XYV}}^n: \!\! P_{XY}=Q_{XY},$ \\ $ \E_{P_{XYV}} [d_e(X,V)] \leq D_e.$ 
\end{tabular}  \\
\hline
\end{tabular}
\label{tableballout}
\end{table}

The eavesdropper then proceeds assuming $g_1(y^n)$ is the correct joint type. S/he randomly chooses a sequence from a set that covers $T_{Q_{X|Y}}(y^n)$. To this end, we associate with each joint type $Q_{XY}$ a joint type $Q_{XYV}$ from $\mathcal{Q}^n(Q_{XY},D_e)$ (cf. Table~\ref{tablesets}), and generate a sequence uniformly at random from $T_{Q_{V|Y}}(y^n),$ where $Q_{V|Y}$ is the conditional probability induced  by $Q_{XYV}$.  We denote the function corresponding to this stage by $g_2: \mathcal{Y}^n \times \mathcal{Q}_{XY}^n \rightarrow \mathcal{V}^n$. Thus, $g(y^n)=g_2(y^n,g_1(y^n))$.

\bigskip
\begin{Remark}
The above strategy does not depend on the specifics of the function $f$ implemented by the primary user, i.e., it only uses the fact that $d(X^n,f(X^n)) \leq D.$ It is also independent of the source statistics.
\end{Remark}
\bigskip

The following lemma lower-bounds the probability that $g_2(y^n,Q_{XY})$ generates a sequence $V^n$ satisfying $d_e(x^n,V^n) \leq D_e$, for a given pair $(x^n,y^n) \in T_{Q_{XY}},$ i.e., assuming the eavesdropper guesses the joint type correctly.
\bigskip
\begin{Lemma} \label{lemmaprob}
Given joint type $Q_{XYV} \in \mathcal{Q}_{XYV}^n$  and $(x^n,y^n) \in T_{Q_{XY}}$, if $V^n$ is chosen uniformly at random from $T_{Q_{V|Y}}(y^n)$, then $ \Pr \left( V^n \in T_{Q_{V|X}}(x^n) \right) \geq c_n 2^{-nI_{Q_{XYV}}(X;V|Y)},$ where $c_n = (n+1)^{ -|\mathcal{X}||\mathcal{Y}||\mathcal{V}|}$.
\end{Lemma}
\begin{proof}
\begin{align*}
\Pr \left( V^n \in T_{Q_{V|X}}(x^n) \right) & \geq \Pr \left( V^n \in T_{Q_{V|X,Y}}(x^n,y^n) \right) \\
& =\frac{\left|T_{Q_{V|X,Y}}(x^n,y^n)\right|}{\left|T_{Q_{V|Y}}(y^n)\right|} \\
& \geq  \frac{c_n 2^{nH(V|X,Y)}}{2^{nH(V|Y)}}=c_n 2^{-nI(X;V|Y)}.
\end{align*}
where the second inequality follows from Lemma 1.2.5 in~\cite{korner}.
\end{proof}
\bigskip

Since the eavesdropper is interested in maximizing this probability, s/he will associate, with each $Q_{XY}$, a joint type achieving the maximum:
\begin{equation} \label{Pmineav}
P^\star_n(Q_{XY}) \in \argmin_{P_{XYV} \in \mathcal{Q}^n(Q_{XY},D_e)} I(X;V|Y).
\end{equation}
Note that this is the same joint type achieving the maximum in~\eqref{eavesbound}.
\bigskip

We can now lower-bound the probability that $x^n \in B_{D_e}(g(y^n))$, for a given pair $(x^n,y^n)$ satisfying $d(x^n,y^n) \leq D$.
\bigskip
\begin{Lemma} \label{lemmaprobfinal}
Given  $(x^n,y^n) \in \mathcal{X}^n \times \mathcal{Y}^n $ satisfying $d(x^n,y^n) \leq D,$ 
$\Pr \left(x^n \in B_{D_e}\big(g(y^n)\big)\right) \geq c'_n 2^{-nI_{P^\star_n(Q_{XY})}(X;V|Y)}$, where $c'_n=(n+1)^{ -|\mathcal{X}||\mathcal{Y}|(|\mathcal{V}|+1)}$, $Q_{XY}=Q_{x^ny^n}$, and 
$g$ is as described above.
\end{Lemma}
\begin{proof}
\begin{align*}
& \Pr(x^n \in B_{D_e}(g(y^n))) \\
& =  \sum_{Q'_{XY} \in \mathcal{Q}_{XY}^n(Q_{y^n},D)} p(g_1(y^n)=Q'_{XY}) \cdot \\
& \qquad \qquad \qquad \qquad \quad ~~ p \left( x^n \in B_{D_e} \Big( g_2(y^n,Q'_{XY} ) \Big) \right) \\
& \geq p(g_1(y^n)=Q_{x^n y^n}) p \left( x^n \in B_{D_e} \Big( g_2(y^n,Q_{x^n y^n} ) \Big) \right) \\
& \geq (n+1)^{-|\mathcal{X}||\mathcal{Y}|}  p \left( x^n \in B_{D_e} \Big( g_2(y^n,Q_{XY} ) \Big) \right) \\
&  \geq (n+1)^{ -|\mathcal{X}||\mathcal{Y}|(|\mathcal{V}|+1)} 2^{-nI_{P^\star_n(Q_{XY})}(X;V|Y)},
\end{align*}
where the last inequality follows from Lemma~\ref{lemmaprob}.
\end{proof}

We now show that the above described scheme indeed achieves the exponent in \eqref{maineq}. Consider any possibly random function $f$ implemented by the primary user (and satisfying the distortion constraint), and denote by $P_f$ the induced joint probability on $(X^n,Y^n)$.  Now, consider the following chain of inequalities.
\begin{align} \label{mainchainconv}
 & \Pr \Big(d_e \Big(X^n, g(f(X^n)) \Big) \leq  D_e \Big) \notag \\
& = \sum_{x^n \in \mathcal{X}^n} \sum_{y^n \in \mathcal{Y}^n}  P(x^n)P_f (y^n | x^n) p \left( x^n \in B_{D_e} \big(g(y^n)\big) \right) \notag \\
& \stackrel{\text{(a)}} \geq c'_n \sum_{x^n \in \mathcal{X}^n} \sum_{y^n \in \mathcal{Y}^n}  P(x^n)P_f (y^n | x^n)2^{-nI_{P^\star_n(Q_{x^ny^n})}(X;V|Y) } \notag \\
&  \geq    c'_n 
 \sum_{x^n \in \mathcal{X}^n}   P(x^n) \sum_{y^n \in \mathcal{Y}^n} P_f (y^n | x^n) \cdot \notag \\
 &  \qquad \qquad \qquad \qquad  \min_{Q_{XY} \in \mathcal{Q}_{\mathcal{XY}}^n(Q_{x^n},D)}   2^{-nI_{P^\star_n(Q_{XY})}(X;V|Y) } \notag \\
& = c'_n \sum_{Q_{X} \in \mathcal{Q}_\mathcal{X}^n} \!\! \sum_{x^n \in T_{Q_X}}  P(x^n)   \cdot \notag \\
& \qquad \qquad \qquad \qquad \min_{ Q_{XY} \in  \mathcal{Q}_{\mathcal{XY}}^n(Q_X,D)  }  2^{-nI_{P^\star_n(Q_{XY})}(X;V|Y) }  \notag \\
& \stackrel{\text{(b)}} = c'_n \sum_{Q_{X} \in \mathcal{Q}_\mathcal{X}^n} \sum_{x^n \in T_{Q_X}} 2^{-n(D(Q_X||P)+H_{Q_X}(X)) } \cdot \notag \\
& \qquad \qquad \qquad \qquad \qquad  2^{-nI_{P^\star_n(Q_{X})}(X;V|Y) }  \notag \\
& \stackrel{\text{(c)}} \geq c'_n (n+1)^{-|\mathcal{X}|} \sum_{Q_{X} \in \mathcal{Q}_\mathcal{X}^n}  2^{-n \left(D(Q_X||P)+ I_{P^\star_n(Q_{X})}(X;V|Y) \right)}  \notag \\
& \geq c'_n (n+1)^{-|\mathcal{X}|} \max_{Q_{X} \in \mathcal{Q}_\mathcal{X}^n}  2^{-n \left(D(Q_X||P)+ I_{P^\star_n(Q_{X})}(X;V|Y) \right)},  
\end{align} 
where
\begin{enumerate}
\item[(a)] follows from Lemma~\ref{lemmaprobfinal}.
\item[(b)] follows from \eqref{Pmineav} and \eqref{Imax}. 
\item[(c)] follows from Lemma 1.2.3 in~\cite{korner} $ \left(|T_{Q_X}| \geq (n+1)^{-|\mathcal{X}|} 2^{nH_{Q_X}(X)} \right).$
\end{enumerate}

Taking the limit as $n$ goes to infinity,  we get
\begin{align} 
 &E^+(P,D,D_e)  \notag \\
 &= \limsup_{n \rightarrow \infty} \max_{\{f_n\}} \min_{\{g_n\}} -\frac{1}{n} \log  \Pr \Big( \! d_e \! \Big(X^n, g_n(f_n(X^n)) \Big) \!\! \leq \!  D_e \! \Big) \notag \\
& \leq \min_{Q} D(Q||P)+R(Q,D,D_e), \label{conveq}
\end{align}
where the last inequality follows from Proposition~\ref{proplimit}.

Combining~\eqref{conveq} and~\eqref{acheq} yields that the limit in~\eqref{defE} exists and is equal to the expression given in~\eqref{maineq}, thus establishing Theorem~\ref{MainThm}.

\section{Lossy Communication for the Shannon Cipher System} \label{commonkey}

We now consider the setup of the Shannon cipher system with lossy communication. More precisely, the transmitter is subject to a rate constraint, and the transmitter and legitimate receiver share common randomness $K \in \mathcal{K}=\{0,1\}^{nr}$, where $r >0$ denotes the rate of the key. $K$ is uniformly distributed over $\mathcal{K}$ and is independent of $X^n$. The transmitter sends a message $M=f(X^n,K)$ to the receiver over a noiseless channel at rate $R$, i.e., $M \in \mathcal{M} = \{0,1\}^{nR}$. The receiver, then, generates $Y^n=h(M,K)$. Both functions $f$ and $h$ are allowed to be stochastic, but must satisfy $\Pr \left( d (X^n, Y^n) > D \right) \leq 2^{-n \alpha}$, for a given reliability exponent $\alpha > 0$.

The message $M$ is overheard by the eavesdropper who knows the statistics of the source and the encoding and decoding functions $f$ and $h$. However, s/he does not have access to the common randomness $K$.

As before, the relevant secrecy metric is the probability of a successful guess, i.e., a guess $V^n=g(M)$ satisfying $d_e(X^n,V^n) \leq  D_e$. The optimal guess is determined, again, by the MAP rule $g_o$.

We assume (A1)-(A3) (given in Section~\ref{nokey}) hold throughout this section. We further assume\footnote{For the primary user's problem to be feasible, it is necessary to have $R \geq \max_{Q:D(Q||P) \leq \alpha} R(Q,D)$. }
\begin{enumerate}
\item[(A4)] $R > R_\alpha := \max_{Q: D(Q||P) \leq \alpha} R(Q,D)$. 
\end{enumerate}  
Let $\overrightarrow{D}=(D,D_e)$ and $\overrightarrow{R}=(R,r)$. For a given DMS $P$, distortion vector $\overrightarrow{D}$, rate vector $\overrightarrow{R}$, and reliability exponent $\alpha$, we denote the optimal exponent by $E(P,\overrightarrow{D},\overrightarrow{R},\alpha)$, i.e.,
\begin{align} 
& E(P,\overrightarrow{D},\overrightarrow{R},\alpha)  =  \notag \\
& \lim_{n \rightarrow \infty} \max_{\{f_n\}} \min_{\{g_n\}} -\frac{1}{n}\log \Pr \Big( d_e \Big(X^n,g_n(f_n(X^n,K)) \Big) \! \leq \! D_e \! \Big), \label{defEr}
\end{align}
where $\{f_n\}$ is restricted to the class of functions ensuring the feasibility of the primary user's problem. Similarly to~\eqref{defEr}, we define $E^-(P,\overrightarrow{D},\overrightarrow{R},\alpha)$ and $E^+(P,\overrightarrow{D},\overrightarrow{R},\alpha)$ using the $\liminf$ and $\limsup$, respectively.

\bigskip
We extend the definition of $R(P_X,D,D_e)$ to account for the rate constraint as follows. For a given distribution $P_X$ satisfying $R(P_X,D) \leq R$,
\begin{align}
\label{eqdefRPRDDe}
R(P_X,R,D,D_e) =  \max_{ \substack{ P_{Y|X}: \\ \E[d(X,Y)] \leq D \\ I(X;Y) \leq R
} } R(P_{XY},D_e).
\end{align}
Extending the properties of $R(P_X,D,D_e)$, we prove the following properties of $R(P_X,R,D,D_e)$ in Appendix~\ref{AppendixpropcontinuityR}. 
\begin{Proposition} \label{propcontinuityR} 
In the following statements, $D \geq D_{\min}$, $D_e \geq D_{e,\min}$, and a given pair $(P_X,R)$ satisfy $R \geq R(P_X,D)$.
\begin{enumerate}
\item[(P4)] For fixed $P_X$, $R(P_X,R,D,D_e)$ is a finite-valued function of $(R,D,D_e)$. Moreover, for fixed $D_e$, $R(P_X,R,D,D_e)$ is continuous in the triple $(P_X,R,D)$ over the set $ \mathcal{S} = \{ (P_X,R,D) : P_X \in \mathcal{P}_{\mathcal{X}}, D \geq D_{\min},R > R(P_X,D) \} $.
\item[(P5)] $R_e(P_X,D_e) - R(P_X,D) \leq R(P_X,R,D,D_e) \leq R(P_X,D,D_e) \leq R_e(P_X,D_e)$.
\end{enumerate}
 
\end{Proposition}
\bigskip

The main result is given by the following theorem.
\begin{Theorem} \label{Thmwithkey}
Under assumptions (A1)-(A4), for any DMS $P$, distortion functions $d$ and $d_e$ with associated distortion levels $D \geq D_{\min}$ and $D_e \geq D_{e,\min}$, corresponding respectively to the primary user and the eavesdropper, and reliability exponent $\alpha$:
\begin{align}
& E(P,\overrightarrow{D},\overrightarrow{R},\alpha) = \notag \\
& \min \! \Big\lbrace \! E_0(P,D_e),   r  + \!\!\! \min_{Q: D(Q||P) \leq \alpha} \!\! D(Q||P) \! + \! R(Q,R,D,D_e) \Big\rbrace. \label{eqthmwithkey}
\end{align}
\end{Theorem}
\bigskip
\begin{Remark}
The minimization over $Q$ is due to the imposition of an exponentially decaying probability of violating the distortion constraint. If we replace it instead by $\Pr \left( d(X^n,Y^n) > D \right) \leq \delta$, for some small $\delta$, then the second term of~\eqref{eqthmwithkey} would collapse to $r+R(P,R,D,D_e)$. 
\end{Remark}
\begin{Remark}
We can recover Theorem~\ref{MainThm} by setting $\alpha=+\infty$, $r=0$, and $R= \log |\mathcal{Y}|$. Weinberger and Merhav's result~\cite[Theorem 1]{merhav2015probmetric} can also be recovered by noting that the leniency assumption implies $R(Q,R,D,D_e)=0$ for all $Q$. Moreover,
for any $D_e > \min_{v \in \mathcal{V}} \E_{P} [d(X,v)]$ and $r > 0$, $E(P,\overrightarrow{D},\overrightarrow{R},\alpha)  >0$. Indeed, the first condition implies $R_e(P,D_e) > 0$, hence $E_0(P,D_e) > 0$. This refines Schieler and Cuff's observation~\cite{cuff2013secrecycausal} that any positive key rate drives the distortion at the eavesdropper to its maximal expected value with high probability. 
\end{Remark}
\bigskip

A straightforward but useful corollary of Theorem~\ref{Thmwithkey} is a necessary and sufficient condition on the key rate for the achievability of the perfect secrecy exponent. In particular,
\begin{align}
& E(P,\overrightarrow{D},\overrightarrow{R},\alpha) = E_0(P,D_e) \text{ if and only if } \notag \\
 & r \geq E_0(P,D_e) - \min_{Q: D(Q||P) \leq \alpha} D(Q||P) + R(Q,R,D,D_e). 
\label{eqCorollary}
\end{align}
Let $r_0$ be the minimum rate needed to achieve $E_0(P,D_e)$. The condition in~\eqref{eqCorollary} is interesting in that it allows $r_0$ to be \emph{strictly} less than $E_0(P,D_e)$, which itself satisfies $E_0(P,D_e) \leq R_e(P,D_e)$. 
\begin{Remark}
One might suspect that $r \geq \max_{Q:D(Q||P)\leq \alpha} R(Q,D)$ is sufficient to achieve $E_0(P,D_e)$, since we can use good rate-distortion codes for each type and the number of available keys is large enough to completely ``hide'' the source sequence within a type class. This is, indeed, true as it implies the condition in~\eqref{eqCorollary}:
\begin{align*}
  D(Q||P)+R_e(Q,D_e) - D(Q||P)-R_e(Q,D_e) + R(Q,D)  \\  = R(Q,D), \\
 \Rightarrow  \min_{Q'} [D(Q'||P)+R_e(Q',D_e)] - D(Q||P)-R_e(Q,D_e) \\  \qquad + R(Q,D)  \leq R(Q,D), \\
 \text{By Property (P5): } \min_{Q'} [D(Q'||P)+R_e(Q',D_e)]  - D(Q||P) \\  \qquad -R(Q,R,D,D_e)  \leq R(Q,D), \\
 \Rightarrow  E_0(P,D_e) + \max_{Q:D(Q||P)\leq \alpha}[-D(Q||P)-R(Q,R,D,D_e)]   \\
 \leq \max_{Q:D(Q||P)\leq \alpha} R(Q,D).
\end{align*}
\end{Remark}
\bigskip

The converse of Theorem~\ref{Thmwithkey} is based on the following analysis. To achieve the second exponent in~\eqref{eqthmwithkey}, the eavesdropper tries to guess the value of the key and then applies the scheme suggested in the previous section. Taking into consideration the rate constraint, the term $R(Q,D,D_e)$ which appears in~\eqref{maineq} is replaced by $R(Q,R,D,D_e)$. Also, taking into account the modified distortion constraint, the minimization over all $Q$'s which appears in~\eqref{maineq} is replaced by a minimization over $Q$'s satisfying $D(Q||P) \leq \alpha$.
The first exponent is the perfect secrecy exponent (given in~\eqref{eqperfectsecrecyexp}), which  the eavesdropper can achieve even in the absence of any observation. The fact that one of these two schemes achieves the optimal exponent implies that the eavesdropper does not benefit from guessing only \emph{part} of the key. Either s/he guesses the entire key correctly and proceeds, or s/he makes a completely blind guess. Interestingly, a similar observation has been made by Schieler and Cuff~\cite{cuff2014henchman} in the context of minimum expected distortion over a list.

To describe the achievability result, it is helpful to rewrite~\eqref{eqthmwithkey} as:
\begin{align} 
& E(P,\overrightarrow{D},\overrightarrow{R},\alpha) = \min \Big\lbrace \min_{Q: D(Q||P) \geq \alpha} \!\! D(Q||P) + R_e(Q,D_e), \notag \\
& \! \min_{Q: D(Q||P) \leq \alpha} \!\!\!\!\! D(Q||P) \! + \! \min\{r \! + \! R(Q,R,D,D_e),R_e(Q,D_e)\} \!  \Big\rbrace.  \label{eqthmwithkeyrewrite}
\end{align}
The primary user will operate as follows. For low-probability types $Q$, particularly $Q$'s with $D(Q||P) > \alpha$, the transmitter will send a dummy message. This is feasible because we allowed some probability of violating the distortion constraint. For such $Q$'s, the eavesdropper receives no information. Therefore, the guessing exponent conditioned on $T_Q$ is given by $R_e(Q,D_e)$, yielding the second term of~\eqref{eqthmwithkeyrewrite}. For $Q$'s satisfying $D(Q||P) \leq \alpha$, let $E(\overrightarrow{D},\overrightarrow{R},Q) = \min\{r+R(Q,R,D,D_e),~R_e(Q,D)\}$. This can be understood as the exponent conditioned on $X^n \in T_Q$. For each such $Q$, we associate a joint type induced by a $P_{Y|X}$ that achieves the maximum in~\eqref{eqdefRPRDDe}. Similarly to Section~\ref{primaryachieve}, we use this joint type to generate a rate-distortion code. This roughly corresponds to the term $R(Q,R,D,D_e)$.
To take advantage of the secret key, we in fact produce $2^{nr}$ such codes, and use the key to randomize the choice of the code, yielding the additional $r$ term. Since the eavesdropper can always guess blindly and achieve the exponent $R_e(Q,D)$, we get $ \min\{r+R(Q,R,D,D_e),~R_e(Q,D)\}$. We will show in Lemma~\ref{lemmaachievability} that such random construction fails to achieve the desired exponent with only \emph{doubly exponentially} small probability. \\

As mentioned, the code construction for each type depends on the conditional $P_{Y|X}$ achieving the maximum in~\eqref{eqdefRPRDDe}. A natural question arises: under what conditions does the optimal test channel (which we will denote by $P^\star_{Y|X}$) achieve that max? One can readily verify that this holds when $R(Q,R,D,D_e)=0$ (e.g., the eavesdropper's constraint is more lenient than that of the legitimate receiver). 
We further investigate this question by considering special cases of Theorem~\ref{Thmwithkey}.

\subsection{Applications of Theorem~\ref{Thmwithkey}}
In the following, assume $\alpha=+\infty$. Hence, $R > \max_{Q} R(Q,D)$.

\subsubsection{Perfect Reconstruction at the Eavesdropper}

Suppose $\mathcal{V} = \mathcal{X}$, and the eavesdropper is required to reconstruct the source sequence perfectly, i.e., the secrecy metric is $\Pr (V^n = X^n)$. In our formulation, this is equivalent to setting  $d_e$ to be the Hamming distance and $D_e$ to 0. Then, for each $Q$, we get
\begin{align*}
R(Q,R,D,0) & = \max_{ \substack{ P_{Y|X}: \\ \E[d(X,Y)] \leq D \\ I(X;Y) \leq R }} R(P_{XY},0) \\
& = \!\! \max_{ \substack{ P_{Y|X}: \\ \E[d(X,Y)] \leq D \\ I(X;Y) \leq R}} \!\! H(X|Y) = H_Q(X) - R(Q,D).
\end{align*}
Note that the maximum is achieved by the optimal test channel, and the exponent is given by 
\begin{align*}
E(P,\overrightarrow{D},\overrightarrow{R}) = \min_Q & D(Q||P) + \\
&  \min \{r + H_Q(X) - R(Q,D), ~H_Q(X) \},
\end{align*}
where we have used the equivalent form~\eqref{eqthmwithkeyrewrite}. Note that, in contrast to $R(Q,R,D,D_e)=0$, this case corresponds to a more lenient constraint at the legitimate receiver, which leads us to our next example.
\bigskip

\subsubsection{Binary Source with Hamming Distortion and $D_e \leq D$} Suppose $\mathcal{X}=\mathcal{Y}=\mathcal{V}=\{0,1\}$, $d$ and $d_e$ are both the Hamming distance, and $D_e \leq D < 1/2$. 
 We prove the following lemma in Appendix~\ref{lemmahammingproof}.
\begin{Lemma} \label{lemmahamming}
If $D_e \leq D < 1/2$,
\begin{align*}
& R(Q,D,D_e) \\
& =  R_e(Q,D_e)-R(Q,D) \\
&  = \begin{cases}
0, &  H(Q) \leq H(D_e), \\
H(Q) - H(D_e),&  H(D_e) \leq H(Q) \leq H(D), \\
H(D)-H(D_e), &  H(Q) \geq H(D).
\end{cases}
\end{align*}
\end{Lemma}
 It follows from property (P5) of Proposition~\ref{propcontinuityR} that
\begin{align*}
R(Q,D,D_e) & =  R_e(Q,D_e)-R(Q,D) \\ \Rightarrow R(Q,R,D,D_e) & = R_e(Q,D_e)-R(Q,D).
\end{align*}
Therefore, the exponent is given by:
\begin{align*}
E(P,\overrightarrow{D},\overrightarrow{R}) = \min\{E_1,E_2,E_3\},
\end{align*}
where 
\begin{align*}
E_1 & = \min \Big\lbrace D(Q||P):  \quad H(Q) \leq H(D_e)\Big\rbrace, \\
E_2 & = \min \Big\lbrace D(Q||P) + H(Q)-H(D_e): \\
& \qquad \qquad \qquad \qquad  H(D_e) \leq H(Q) \leq H(D)\Big\rbrace, \intertext{and} 
E_3 & = \min \Big\lbrace D(Q||P) + \\ 
& 
 \qquad \min\{r+H(D)-H(D_e), H(Q)-H(D_e)\}: \\
 & \qquad \qquad   H(Q) \geq H(D) \Big\rbrace.
\end{align*}
%\begin{align*}
%& E(P,\overrightarrow{D},\overrightarrow{R}) = \\ & \min 
%\begin{cases}
%\min \Big\lbrace D(Q||P):  \quad H(Q) \leq H(D_e)\Big\rbrace, \\
%\min \Big\lbrace D(Q||P) + H(Q)-H(D_e): \\  \qquad   H(D_e) \leq H(Q) \leq H(D)\Big\rbrace, \\
%\min \Big\lbrace D(Q||P) + \\
% \quad \min\{r+H(D)-H(D_e), H(Q)-H(D_e)\}:  \\ 
% \qquad H(Q) \geq H(D) \Big\rbrace.
%\end{cases}
%\end{align*}
If $X \sim$ Ber($1/2$), then $D(Q||P)= 1-H(Q)$, and the minima corresponding to the first two cases reduce to $1-H(D_e)$. The third minimum can be computed as follows:
\begin{align*}
%\min_{\substack{Q: \\ H(Q) \geq H(D)}} 1-H(Q)+H(D)-H(D_e) + \min\{r, H(Q)-H(D)\}  
& \min_{\substack{Q: \\ H(Q) \geq H(D)}} 1-H(D_e) + \min \{r+H(D)-H(Q), 0\}  \\
& \quad = 1-H(D_e) + \min \{r+H(D)-1,0\}. 
\end{align*}
Therefore, 
\begin{equation*}
E(P,\overrightarrow{D},\overrightarrow{R}) = \begin{cases}
1-H(D_e), & r \geq 1-H(D),\\
r+H(D)-H(D_e), & r < 1-H(D).
\end{cases}
\end{equation*} 

The resulting expression when $r < 1-H(D)$ admits a simple geometric explanation, shown in Figure~\ref{fighammingballs} below. 
\begin{figure}[htp]
\centering
\includegraphics[scale=0.6]{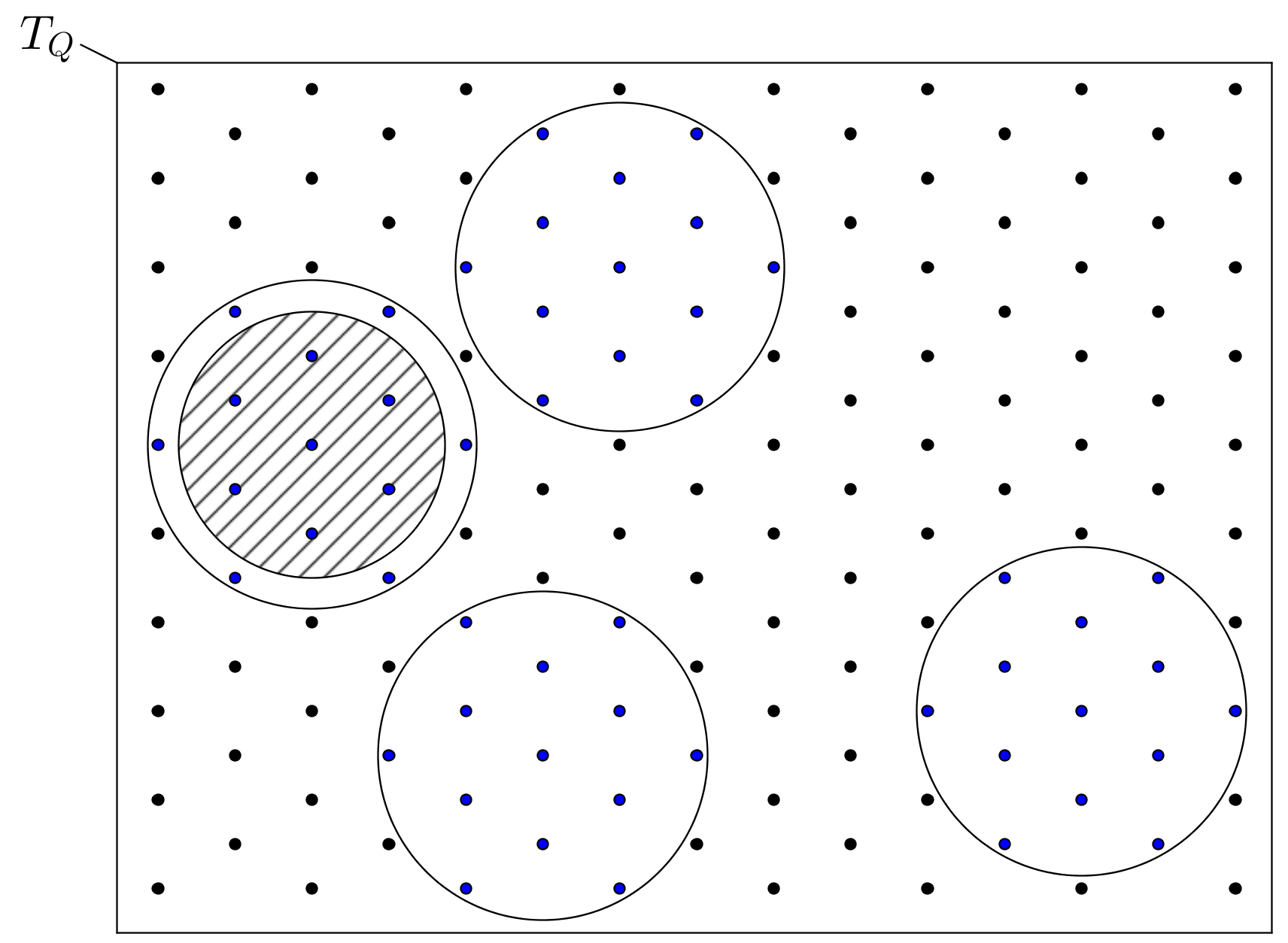}
\caption{The dots represent sequences in a type class $T_Q$. Each of the $2^{nr}$ non-dashed circles represents a Hamming-distortion ball of radius $D$, corresponding to a possible reconstruction at the legitimate receiver. Thus, dots within the circle (in blue) represent candidate source sequences. The dashed circle represents the distortion ball of radius $D_e$ around the eavesdropper's reconstruction, and it fits entirely in a non-dashed circle.} \label{fighammingballs}
\end{figure}
Upon observing the public message, the candidate source sequences are clustered into $2^{nr}$ balls. Each ball corresponds to a possible value of the key $K$, and has volume $2^{nH(D)}$ since it is the pre-image of a possible reconstruction at the legitimate receiver. For the eavesdropper, the maximum volume of the ball that s/he can generate to ``engulf'' candidate sequences is $2^{nH(D_e)}$. Due to the structure of Hamming distortion, this maximally-sized ball can fit entirely into any one of the clusters, so that the probability of a successful guess is $2^{nH(D_e)}2^{-n(r+H(D))}$. 
Note that the geometric interpretation assumed that we are using good rate-distortion codes (to get pre-images of volume $2^{nH(D)}$). The described structure is also reminiscent of successive refinement~\cite{SuccessiveRefinement}. These can be explained by the following lemma.
\begin{Lemma} \label{lemmasuccessive}
If $R(Q,D,D_e) = R_e(Q,D_e)-R(Q,D)$, then the optimal test channel $P^\star_{Y|X}$ achieves the maximum in~\eqref{Rmaxmin}. Moreover, $Q$ is successively refinable from $D$ to $D_e$.
\end{Lemma}
\begin{IEEEproof}
Consider the proof of the lower bound in (P3) of Proposition~\ref{continuity}. $R(Q,D,D_e) = R_e(Q,D_e)-R(Q,D)$ implies that~\eqref{eqp3proof1} is an equality. Hence, $P^\star_{Y|X}$ achieves the maximum. Moreover,~\eqref{eqp3proof2} becomes an equality. Let $P^{(1)}_{V|XY}$ be the minimizer in~\eqref{eqp3proof1}, and $P^{(2)}_{V|XY}$ the minimizer in~\eqref{eqp3proof2}. Then, $H_{P^\star_{XY} P^{(1)}_{V|XY}} (X|V) \leq H_{P^\star_{XY} P^{(2)}_{V|XY}} (X|V) = H_{P^\star_{XY} P^{(1)}_{V|XY}} (X|V,Y) \leq H_{P^\star_{XY} P^{(1)}_{V|XY}} (X|V).$ Therefore, $P^{(1)}_{V|XY}$ satisfies $\E[d_e(X,V)] \leq D_e$ and  $H_{P^\star_{XY} P^{(1)}_{V|XY}} (X|V,Y) = H_{Q P^{(1)}_{V|X}} (X|V)$, i.e. the Markov chain $X-V-Y$ holds. Finally, note that $I(X;V)=I(X;(V,Y))=I(X;Y)+I(X;V|Y)=R_e(Q,D_e)$, implying that $Q$ is successively refinable from $D$ to $D_e$.
\end{IEEEproof}

\subsection{Achievability Proof} \label{Secachkey}

We show that $E^-(P,\overrightarrow{D},\overrightarrow{R},\alpha) \geq \min \{ E_0(P,D), ~ r  + \min_{Q: D(Q||P) \leq \alpha} D(Q||P) + R(Q,R,D,D_e) \}$ by demonstrating an encoding-decoding strategy for the primary user that achieves the given exponent.

As before, the primary user will operate on the source sequences on a type-by-type basis. The result is driven by the following lemma, which is based on the analysis of Schieler and Cuff~\cite{cuff2014henchman} and the proof of which is given in Appendix~\ref{mainlemmaproof}.

\bigskip

\begin{Lemma} \label{lemmaachievability}
Let $\epsilon \! > \!  0, n  \! \in \!  \mathbb{N}, Q_{XY} \!  \in  \! \mathcal{Q}_{\mathcal{XY}}^n$ be given.  Let $N$ be an integer such that $2^{n(I_{Q_{XY}}(X;Y)+2\epsilon/3)} \! \leq \! N \! \leq \! 2^{n(I_{Q_{XY}}(X;Y)+\epsilon)} $.  Generate a code $\mathcal{C}^n=(Y_1^n,Y_2^n,\dots,Y_N^n)$ by choosing $N$ elements independently and uniformly at random from $T_{Q_Y}$. \begin{enumerate}
\item {\bf Covering:} For $x^n \in T_{Q_X}$, define
\begin{align}
\label{eqdefNxn}
\mathcal{C}(x^n) & =  \{ m \in [N]: (x^n,Y_m^n) \in T_{Q_{XY}} \}, \\
 N_{x^n} & = |\mathcal{C}(x^n)|, 
\end{align} and the event
\begin{align}
\mathcal{E} = \big\lbrace \mathcal{C}^n: & \text{ there exists } x^n \in T_{Q_X} \text{ such that } N_{x^n} = 0 \notag \\
& \text{ or } N_{x^n} > 2^{2n\epsilon} \big\rbrace. \label{eqdefEv} 
\end{align}
Then, there exists $n_1(\epsilon,|\mathcal{X}|,|\mathcal{Y}|)$ (independent of $Q_{XY}$) such that, for all $n \geq n_1$,
\begin{align}
\label{eqlemmaach1}
\Pr(\mathcal{E}) \leq e^{-2^{n \epsilon/7}}.
\end{align}
\item {\bf Guessing---single code:} Suppose $X^n \sim \mathrm{Unif}(T_{Q_X})$ and $\mathcal{C}^n \notin \mathcal{E}$. Let $P_{M|X^n}^{\mathcal{C}}$ be as follows. Given $x^n$, $M$ is chosen uniformly at random from $\mathcal{C}(x^n)$. Then, for all $n \geq n_1$, for all $v^n \in \mathcal{V}^n$ and all $m \in [N]$,
\begin{align}
\label{eqlemmaach2}
\Pr(d_e(X^n,v^n) \! \leq \! D_e |M \!\! = \! m,\mathcal{C}^n) \! \leq  \! 2^{-n( \! R(Q_{XY},D_e \!) - 4\epsilon)},
\end{align}
and \begin{align}
 \E [ \Pr(d_e(X^n,v^n) \leq D_e |M=m,\mathcal{C}^n) | \mathcal{E}^c ] \leq \notag \\
 2^{-n(R_e(Q_X,D_e)-4 \epsilon)}, \label{eqlemmaach3}
\end{align} where the probabilities are computed with respected to the randomness in $P_{M|X^n}^\mathcal{C}$, and the expectation with respect to the distribution of the code $\mathcal{C}^n$.
\item {\bf Guessing---multiple codes:}  Let $K$ be uniform over $\mathcal{K}=[2^{nr}]$, $r>0$, and independent of $X^n$ (which is uniform over $T_{Q_X}$). For each $k \in [2^{nr}]$, generate $\mathcal{C}^n_k$ as described above. Define $P_{M|X^n,K}^{\{\mathcal{C}^n_k \}_k}$ as follows. Given $k$, if $C_k \in \mathcal{E}$, then $M$ is chosen uniformly at random from $[N]$ independently of $X^n$. If  $\mathcal{C}^n_k \notin \mathcal{E}$, then $M$ is chosen uniformly at random from $\mathcal{C}_k(X^n)$. %Let $M \in [N]$ be the index of $Y^n$ in the corresponding $\mathcal{C}^n_k$. 
Let, 
\begin{align}
\tilde{\mathcal{E}} & =   \bigg\{ \{   \mathcal{C}_k \}_{k=1}^{2^{nr}}:    \text{ there exists } m \in [N] \text{ such that}  \notag \\
& \left. \max_{v^n \in \mathcal{V}^n} \Pr \left(d_e(X^n,v^n) \leq D_e |M=m, \{\mathcal{C}_k^n\}_{k=1}^{2^{nr}} \right) > \right. \notag \\
& \left.  2^{-n (\min \{ R_e(Q_X,D_e), r+R(Q_{XY},D_e)\}-8\epsilon)} \right\rbrace,  \label{eqdefEvg}
\end{align}  where the probability is computed with respect to $P_{M|X^n,K}^{\{ \mathcal{C}^n_k\}_k}$. Then, for all $n \geq n_1$,
\begin{align}
\label{eqlemmaach4}
\Pr ( \tilde{\mathcal{E}} ) \leq  e^{-2^{n\epsilon/9}},
\end{align}
where the probability is computed with respected to the distributions of the codes $\{\mathcal{C}^n_k\}_k$. \hfill $\square$
\end{enumerate}
\end{Lemma}
\bigskip

The first part of the Lemma asserts that if we generate the codebook randomly, then each $x^n \in T_{Q_{X}}$ will be covered by a small number of codewords (the probability that this event does not occur is doubly exponentially small). Therefore, if we encode $x^n$ by choosing a codeword uniformly at random from its cover, the induced $P_{X^n|Y^n}(.|y^n)$ will be roughly uniform over the set $T_{Q_{X|Y}}(y^n)=\{ x^n: (x^n,y^n) \in T_{Q_{XY}}\}$. Consequently, given a codeword index $m$ and $v^n \in \mathcal{V}^n$, the second part bounds the probability that $v^n$ covers $X^n$, and also bounds the expectation (over the choice of the codebook) of that probability. 

Finally, the third part considers generating $2^{nr}$ codebooks and the induced distribution $P_{X^n|M}$, where $M$ is the index of a chosen codeword. This distribution roughly corresponds to generating $2^{nr}$ elements uniformly at random from $T_{Q_Y}$, revealing the chosen elements to the adversary, then choosing one of them uniformly at random and generating $X^n$ uniformly at random from $T_{Q_{X|Y}} (Y^n)$. This setup is similar to the one studied by Schieler and Cuff~\cite[Theorem 4]{cuff2014henchman}. Equation~\eqref{eqlemmaach4} states that, for most realizations of the codebooks, the probability that the adversary generates a successful guess, given a codeword  index, is upper-bounded by $2^{-n (\min \{ R_e(Q_X,D_e), r+R(Q_{XY},D_e)\}}$. The implication is that the best the adversary could do is 1) either ignore the index and guess $X^n$ blindly, 2) or guess which codebook is being used (i.e., guess the value of the key $K$) and use the scheme suggested in the previous section.\\

Now, fix $\delta > 0$ such that $R_{\alpha+\delta} = \max_{Q: D(Q||P) \leq \alpha+\delta} R(Q,D) < R$. Note that such $\delta$ exists since $\lim_{\delta \rightarrow 0} R_{\alpha+\delta} = R_{\alpha}$ (which follows from Proposition~\ref{propcontinuityD} in Appendix~\ref{continuityproof} and the fact that $D(Q||P)$ is convex).  Fix $R'$ such that $R_{\alpha+\delta} < R' < R$, and $\epsilon > 0$ such that $\epsilon < R-R'$. Let \begin{align} \label{eqdefQXalpha}
\mathcal{Q}_\mathcal{X}^n (\alpha,\delta) = \{Q \in \mathcal{Q}_{\mathcal{X}}^n: D(Q||P) \leq \alpha+\delta \}.
\end{align}

Let $n$ be large as given by Lemma~\ref{lemmaachievability}. For each type $Q_X \in \mathcal{Q}_\mathcal{X}^n (\alpha,\delta)$, we associate a joint type $Q_{XY}$ and generate $2^{nr}$ codebooks $\{\mathcal{C}_k\}_{k=1}^{2^{nr}} \in \tilde{\mathcal{E}}^c $ where the size of each codebook is upper-bounded by $2^{n(I_{Q_{XY}} (X;Y) + \epsilon)}$
 (the existence of such codes follows from~\eqref{eqlemmaach4}). Since the primary user wants to minimize the probability of a successful guess by the eavesdropper, but must also satisfy a rate constraint, the associated type is chosen as follows:
 \begin{align} \label{eqdefQstarRQX}
 Q^\star_{R'} (Q_X) \in \argmax_{ \substack{ Q_{XY} \in \mathcal{Q}_\mathcal{XY}^n (Q_X,D): \\ I_{Q_{XY}} (X;Y) \leq R'}} R(Q_{XY},D_e).
 \end{align}
 
 The encoding function $f$ is as follows. Given a source sequence $x^n$ satisfying $Q_{x^n} \in \mathcal{Q}_\mathcal{X}^n (\alpha,\delta)$, and a realization of the key $k$, a reconstruction sequence is chosen uniformly at random from $\mathcal{C}_{k} (x^n)$ (cf.~\eqref{eqdefNxn}). The associated message is then given by:
 \begin{itemize}
 \item $\lceil \log |\mathcal{Q}_{\mathcal{X}}^n | \rceil$ bits to describe $Q_X$.
 \item $\lceil \log | \mathcal{C}_k^n| \rceil$ bits to describe the index of the reconstruction.
 \end{itemize}
The legitimate receiver uses the first part of the message and the key to determine which codebook is being used, and then uses the second part of the message to recover the reconstruction $Y^n$. Finally, all sequences $x^n$ such that $Q_{x^n} \notin  \mathcal{Q}_\mathcal{X}^n (\alpha,\delta)$ are mapped to an arbitrary message $m_0$.

\begin{Remark}
One can check that this encoding is feasible by noting that the required number of bits satisfy:
\begin{align*}
\lceil \log |\mathcal{Q}_{\mathcal{X}}^n | \rceil +\lceil \log | \mathcal{C}_k^n| \rceil & \leq \\
 |\mathcal{X}| \log(n+1) +1 + n \left(I_{Q^\star_{R'} (Q_X)}(X;Y) + \epsilon \right) + 2 & \leq \\
  n \left( R'+\epsilon+ \frac{|\mathcal{X}|}{n} \log(n+1) +  \frac{3}{n} \right) & < nR,
\end{align*}
for $n$ large enough. Moreover, it satisfies the excess distortion probability constraint, since 
\begin{align*}
 \Pr(d(X^n,Y^n) > D) & \leq \sum_{Q_X \notin \mathcal{Q}_n(\alpha,\delta)} P(Q_X) \\
 & \leq \sum_{Q_X \notin \mathcal{Q}_n(\alpha,\delta)}  2^{-nD(Q_X||P)}  \\ & \leq (n+1)^{|\mathcal{X}|} 2^{-n(\alpha+\delta)}  < 2^{-n\alpha},
\end{align*}
where the last inequality holds for large enough $n$.
\end{Remark}
\bigskip

To analyze the performance of the eavesdropper, note that when s/he observes a message $m \neq m_0$, then the induced distribution $P_{X^n|M=m}$ is exactly the setup studied in part three of Lemma~\ref{lemmaachievability}. Indeed, the message $m$ indicates the type of the transmitted sequence and the index of the reconstruction (among $2^{nr}$ possible codebooks). For $m=m_0$, i.e., for sequences of type outside $\mathcal{Q}_n(\alpha,\delta)$, the performance can still be analyzed in light of Lemma~\ref{lemmaachievability} by considering the associated $Q_{XY}$ to be of the form $Q_X Q_Y$ (i.e., $X$ and $Y$ are independent), in which case $\min\{R_e(Q_X,D_e),r+R(Q_{XY},D_e)\} = \min\{R_e(Q_X,D_e),r+R_e(Q_X,D_e)\} = R_e(Q_X,D_e)$. Now, consider the following chain of inequalities.
\begin{align}
& \Pr \left( d_e \Big(X^n, g_0 (f(X^n,K)) \Big) \leq D_e \right) \notag \\
& = \sum_{Q_X \in \mathcal{Q}_{\mathcal{X}}^n} \sum_{x^n \in T_{Q_X}} \sum_{m \in \mathcal{M}} P(x^n)P_f(m|x^n) \cdot \notag \\
& \qquad \qquad \qquad \qquad \qquad   \mathbf{1} \left\lbrace x^n \in B_{D_e} \left( g_0 (m) \right) \right\rbrace \notag \\
& =  \sum_{Q_X \in \mathcal{Q}_{\mathcal{X}}^n} P(Q_X) \sum_{m \in \mathcal{M}} \sum_{x^n \in T_{Q_X}} P_f(m|T_{Q_X}) P_f(x^n|m,T_{Q_X}) \cdot \notag \\
& \qquad \qquad \qquad \qquad  \qquad \mathbf{1} \left\lbrace x^n \in B_{D_e} \left( g_0 (m) \right) \right\rbrace \notag \\
& =  \sum_{Q_X \in \mathcal{Q}_{\mathcal{X}}^n} P(Q_X) \sum_{m \in \mathcal{M}} P_f(m|T_{Q_X}) \cdot \notag \\
& \qquad \qquad \qquad \qquad \quad  \max_{ v^n \in \mathcal{V}^n} P_f \left( d_e(X^n,v^n) \leq D_e | m, T_{Q_X} \right) \notag \\
&  \stackrel{\text{(a)}} \leq \sum_{Q_X \in \mathcal{Q}_{\mathcal{X}}^n} P(Q_X) \sum_{m \in \mathcal{M}} P_f(m|T_{Q_X}) \cdot \notag \\
& \qquad \qquad \qquad \qquad 2^{-n (\min \{ R_e(Q_X,D_e), r+R({Q^\star_{R'} (Q_X)},D_e)\}-8\epsilon)} \notag \\
& \leq \!\! \sum_{Q_X \in \mathcal{Q}_{\mathcal{X}}^n} \!\! 2^{-n (D(Q_X||P) + \min \{ R_e(Q_X,D_e), r+R(Q^\star_{R'} (Q_X),D_e)\}-8\epsilon)} \notag \\
& \leq (n+1)^{|\mathcal{X}|}  \max_{Q_X \in \mathcal{Q}_{\mathcal{X}}^n}  \exp \Big\lbrace
-n \left(D(Q_X||P) + \right. \notag \\
& \qquad \qquad ~~  \left.  \min \{ R_e(Q_X,D_e), r+R({Q^\star_{R'} (Q_X)},D_e)  \}-8\epsilon \right) \!\! \Big\rbrace \notag \\
& = (n+1)^{|\mathcal{X}|}  \exp \bigg( \!\! -n  \min \left\lbrace 
\min_{Q_X \in \mathcal{Q}_{\mathcal{X}}^n(\alpha,\delta)} D(Q_X||P) + \right. \notag \\
& \qquad \qquad \qquad \quad  \min \{ R_e(Q_X,D_e), r+R({Q^\star_{R'} (Q_X)},D_e)  \} ,  \notag \\
& \qquad \qquad   \left. 
~~ \min_{Q_X \notin \mathcal{Q}_{\mathcal{X}}^n(\alpha,\delta)} D(Q_X||P)+ R_e(Q_X,D_e) 
\right\rbrace + 8\epsilon n  \bigg) \notag \\
& = (n+1)^{|\mathcal{X}|} \exp \left( -n  \min \left\lbrace 
\min_{Q_X \in \mathcal{Q}_{\mathcal{X}}^n(\alpha,\delta)} D(Q_X||P)+ r+ \right. \right. \notag \\
& \qquad \qquad \qquad \qquad \qquad \qquad \qquad R({Q^\star_{R'} (Q_X)},D_e)  , \notag \\
& \qquad \left. \left.   \min_{Q_X \in \mathcal{Q}_{\mathcal{X}}^n} D(Q_X||P)+ R_e(Q_X,D_e) 
\right\rbrace + 8\epsilon n  \right), \label{eqboundguesswithkey}
\end{align}
where (a) follows from~\eqref{eqdefEvg} of Lemma~\ref{lemmaachievability} and the fact that the codebooks $\{\mathcal{C}_k\}_{k=1}^{2^{nr}} \notin \tilde{\mathcal{E}}$ by construction. Therefore, 
\begin{align*}
& E^-(P,\overrightarrow{D},\overrightarrow{R}, \alpha ) \\
& = \liminf_{n \rightarrow \infty} \max_{ \{f_n\} } \min_{ \{g_n\}} \! -\frac{1}{n} \log \Pr \!\! \left( \! d_e \! \Big( \! X^n \! , g_n (f_n(X^n,K)) \! \Big) \! \leq \! D_e \! \right) \\
& \geq \min \bigg\lbrace E_0(P,D),  \\
& \qquad \quad \left. r  + \!\! \min_{Q: D(Q||P) \leq \alpha+\delta} \!\! D(Q||P) + R(Q,R',D,D_e) \!\! \right\rbrace \!\! - 8 \epsilon,
\end{align*} where the inequality follows from the following proposition, the proof of which is given in Appendix~\ref{AppendixproplimitR}.
\begin{Proposition} \label{proplimitR}
\begin{align*}
 \lim_{n \rightarrow \infty} & \min_{Q_X \in \mathcal{Q}_{\mathcal{X}}^n(\alpha,\delta)} D(Q_X||P)+R({Q^\star_{R'} (Q_X)},D_e) = \\
& \min_{Q: D(Q||P) \leq \alpha+ \delta} D(Q||P) + R(Q,R',D,D_e).
\end{align*}
\end{Proposition}
Now, note that $\epsilon$ is arbitrary, and 
\begin{align} 
\lim_{R' \rightarrow R} & \min_{Q: D(Q||P) \leq \alpha+\delta} D(Q||P) + R(Q,R',D,D_e) = \notag \\
&    \min_{Q: D(Q||P) \leq \alpha+\delta} D(Q||P) + R(Q,R,D,D_e), \label{eqlimminR}
\end{align}
since $R(Q,\tilde{R},D,D_e)$ is uniformly continuous in $(Q,\tilde{R})$ over the set $\{(Q,\tilde{R}): D(Q||P) \leq \alpha+\delta, R' \leq \tilde{R} \leq R \}$ by Proposition~\ref{propcontinuityR}. Finally, it follows from Proposition~\ref{propcontinuityR} and Proposition~\ref{propcontinuityD} (to follow in Appendix~\ref{continuityproof}) that
\begin{align} 
\lim_{\delta \rightarrow 0} & \min_{Q: D(Q||P) \leq \alpha+\delta} D(Q||P) + R(Q,R,D,D_e) = \notag \\
&  \min_{Q: D(Q||P) \leq \alpha} D(Q||P) + R(Q,R,D,D_e) \label{eqlimmindelta}
\end{align}
As such,
%\begin{align*}
%& E^-(P,\overrightarrow{D},\overrightarrow{R}, \alpha ) \\
%& = \liminf_{n \rightarrow \infty} \max_{\{f_n\}} \min_{\{g_n\}} -\frac{1}{n}\log \Pr \Big( d_e \Big(X^n,g_n(f_n(X^n,K)) \Big) \leq D_e \Big) \notag \\  
%& \geq \min \left\lbrace E_0(P,D), ~ r  + \!\! \min_{Q: D(Q||P) \leq \alpha} D(Q||P) + R(Q,R,D,D_e) \right\rbrace.
%\end{align*}
\begin{align*}
& E^-(P,\overrightarrow{D},\overrightarrow{R}, \alpha ) \geq \min \bigg\lbrace E_0(P,D),  \\
& \qquad \qquad \left. r  + \!\! \min_{Q: D(Q||P) \leq \alpha} \!\! D(Q||P) + R(Q,R,D,D_e)  \right\rbrace.
\end{align*}

\subsection{Converse Proof}\label{Secconvkey}

We now prove that $E^+(P,\overrightarrow{D},\overrightarrow{R},\alpha) \leq \min \{r+ \min_{Q:D(Q||P) \leq \alpha} D(Q||P)+R(Q,R,D,D_E), E_0(P,D_e)\}$. We have already shown, following Theorem~\ref{MainThm}, that the perfect secrecy exponent $E_0(P,D_e)$ is achievable by the eavesdropper even in the absence of any observation. It follows immediately that
\begin{equation} \label{eqconvkey1}
E^+(P,\overrightarrow{D},\overrightarrow{R},\alpha) \leq E_0(P,D_e).
\end{equation}

So we only need to demonstrate a strategy that achieves the first exponent. The strategy is based on the one suggested in Section~\ref{primaryconv}. We will add an initial stage in which the eavesdropper tries to guess the value of $K$, by choosing an element uniformly at random from $\{1,2,\cdots,2^{nr} \}$. The eavesdropper's guess, denoted by $\tilde{K}$, is equal to $K$ with probability $2^{-nr}$ (This will correspond to the $r$ term in~\eqref{eqthmwithkey}). Then, s/he generates $\tilde{Y}^n=h(M,\tilde{K})$. 
Next, the eavesdropper implements the same stages suggested in Section~\ref{primaryconv}, where $\tilde{Y}^n$ plays the role of $Y^n$. We denote the  strategy by $g^\prime$. 

\begin{Remark}
If $h$ is stochastic, it can be replaced by a deterministic $h$ that still satisfies the reliability constraint. Since this does not change the conditional $P_{X^n|M}$, we can assume, without loss of generality, that $h$ is deterministic. 
\end{Remark}

Now, consider any functions $f$ and $h$ implemented by the primary user (and satisfying the distortion constraint). Let $P_f$ denote the induced joint probability  of $(X^n,M,K)$, and $P_K$ denote the distribution of $K$. To analyze the performance of $g'$, note that, unlike Section~\ref{nokey}, not every realization of $\tilde{Y}^n$ necessarily satisfies the distortion constraint. To that end, define
\begin{align}
& \mathcal{M}_D(x^n,k)  = \{ m \in \mathcal{M}: d(x^n,h(m,k)) \leq D \}, \notag \\
&  \qquad \qquad \qquad \qquad \qquad \qquad \quad  x^n \in \mathcal{X}^n, ~ k \in \mathcal{K}, \label{eqdefMD} \\
& \text{and } \mathcal{A}  = \{ (x^n,y^n) \in \mathcal{X}^n \times \mathcal{Y}^n: d(x^n,y^n) > D \}. \label{eqdefA}
\end{align}
The distortion constraint implies that 
\begin{align}
\label{eqboundPrA}
P_f( \mathcal{A}  ) \leq  2^{-n \alpha}.  
\end{align} 
%This is reflected in the fact that the minimization in the second term of~\eqref{eqthmwithkey} excludes $Q$'s satisfying $D(Q||P)> \alpha$. 
Moreover, the analysis of  $g'$ should take into account the rate constraint $R$. 
The following Lemma by Weissman and Ordentlich~\cite{WeissmanEmpiricalRateConstrained} will be instrumental. 
\begin{Lemma}[{\cite[Lemma 3]{WeissmanEmpiricalRateConstrained}}] \label{LemmaWeissman}
Let $Y^n(.)$ be an $n$-block code of rate $\leq R$. Then, for every $Q \in \mathcal{Q}_{\mathcal{X}}^n$ and $\eta > 0$, if $X^n$ is uniformly distributed over $T_Q$, 
\begin{align} \label{eqLemmaWeissman}
\Pr  \left( \left\lbrace x^n \in T_Q: I_{Q_{x^n Y^n(x^n)}} (X;Y) > R + \eta \right\rbrace \right) \leq \notag \\ (n+1)^{|\mathcal{X}||\mathcal{Y}|+|\mathcal{X}|} 2^{- n\eta}.
\end{align}
\end{Lemma}
\begin{Remark} This is not the exact statement found in~\cite{WeissmanEmpiricalRateConstrained}, but it is a straightforward modification. 
 \end{Remark}
So, define for every $\eta > 0$, $x^n \in \mathcal{X}^n$, and $k \in \mathcal{K}$,
\begin{align}
\mathcal{M}_R (x^n,k,\eta) & = \{  m \in \mathcal{M}: I_{Q_{x^ny^n}} (X;Y) \leq R+\eta,  \notag \\
&  \qquad \text{where } y^n  =h(m,k) \}, \label{eqdefMR} \\
\mathcal{M}_{D,R} (x^n,k,\eta) & = \mathcal{M}_D (x^n,k) \cap \mathcal{M}_R (x^n,k,\eta),  \label{eqdefMDR} \\
\text{and }  \mathcal{B}(\eta) & = \big\lbrace (x^n,y^n) \in \mathcal{X}^n \times \mathcal{Y}^n: \notag \\
& \qquad I_{Q_{x^ny^n}} (X;Y) > R + \eta \big\rbrace. \label{eqdefB} 
\end{align} 
 
 Finally, fix $\epsilon > 0$, $\delta >0 $ and $\eta > 0$, and consider the following chain of inequalities.
\begin{align}
& \Pr \Big( d_e \Big(X^n,g^\prime(f(X^n,K)) \Big) \leq  D_e \Big)\notag \\
& =  \sum_{x^n \in \mathcal{X}^n} \sum_{k \in \mathcal{K}}  \sum_{m \in \mathcal{M}} P(x^n)P_K(k) P_f(m|x^n,k) \cdot \notag \\
& \qquad \qquad \qquad \qquad P_f \left(x^n \in B_{D_e} \Big( g^\prime (m)  \Big) \right) \notag \\
& \geq  \sum_{x^n \in \mathcal{X}^n} \sum_{k \in \mathcal{K}}  \sum_{ m \in  \mathcal{M}_{D,R} (x^n,k,\eta)} P(x^n)P_K(k) P_f(m|x^n,k) \cdot \notag \\
& \qquad \qquad \qquad \qquad  P_f \left(x^n \in B_{D_e} \Big( g^\prime (m)  \Big) \right) \notag \\
& =  \sum_{x^n \in \mathcal{X}^n} \sum_{k \in \mathcal{K}} \sum_{ m \in  \mathcal{M}_{D,R} (x^n,k,\eta)} P(x^n)P_K(k) P_f(m|x^n,k) \cdot \notag \\
& \qquad \qquad \sum_{\tilde{k} \in \mathcal{K} } P_K \left(\tilde{K}  = \tilde{k} \right)  P_f \left(x^n \in B_{D_e} \left( g ( h(m,\tilde{k}) )  \right)  \right) \notag \\
& \geq 2^{-nr} \sum_{x^n \in \mathcal{X}^n} \sum_{k \in \mathcal{K}} \sum_{ \substack{m \in  \mathcal{M}_{D,R} (x^n,k,\eta)}} \!\!\!\!\!\! P(x^n)P_K(k) P_f(m|x^n,k)\cdot \notag \\
& \qquad \qquad \qquad \qquad \qquad \qquad    P_f \left(x^n \in B_{D_e} \left( g ( h(m,k) )  \right)  \right) \notag \\
& \stackrel{\text{(a)}} \geq c'_n 2^{-nr}  \!\!\! \sum_{x^n \in \mathcal{X}^n} \sum_{k \in \mathcal{K}} \sum_{ \substack{m \in   \mathcal{M}_{D,R} (x^n,k,\eta)} } \!\!\!\!\!\!  P(x^n)P_K(k) P_f(m|x^n,k) \cdot \notag \\
& \qquad \qquad \qquad \qquad \qquad \qquad   2^{-n  I_{P^\star_n (Q_{x^n {y}^n})} (X;V|Y)} \notag \\
& \stackrel{\text{(b)}}\geq  c'_n 2^{-nr} \sum_{Q_X \in \mathcal{Q}_{\mathcal{X}}^n} \sum_{x^n \in T_Q} \sum_{k \in \mathcal{K}} \sum_{ m \in  \mathcal{M}_{D,R} (x^n,k,\eta)} P(x^n)P_K(k) \cdot \notag \\
& \qquad \qquad P_f(m|x^n,k) \min_{ \substack{Q_{XY} \in \mathcal{Q}_{\mathcal{X}\mathcal{Y}}^n(Q_{X},D): \\ I_{Q_{XY}} (X;Y) \leq R+\eta}} 2^{-n  I_{P^\star_n (Q_{XY})} (X;V|Y)} \notag \\
& =  c'_n 2^{-nr} \sum_{Q_X \in \mathcal{Q}_{\mathcal{X}}^n} \min_{ \substack{Q_{XY} \in \mathcal{Q}_{\mathcal{X}\mathcal{Y}}^n(Q_{X},D): \\ I_{Q_{XY}} (X;Y) \leq R+\eta}} 2^{-n  I_{P^\star_n (Q_{XY})} (X;V|Y)} \cdot \notag \\
& \qquad \qquad  \qquad \qquad  P_f( \mathcal{A}^c \cap \mathcal{B}^c(\eta) \cap T_Q) \notag \\
& \stackrel{\text{(c)}} \geq c'_n 2^{-nr} \!\!\! \sum_{Q_X \in \mathcal{Q}_{\mathcal{X}}^n(\alpha,-\delta)}   \min_{ \substack{Q_{XY} \in \mathcal{Q}_{\mathcal{X}\mathcal{Y}}^n(Q_{X},D): \\ I_{Q_{XY}} (X;Y) \leq R+\eta}} 2^{-n  (R(Q_{XY},D_e)+ \epsilon) }  \cdot \notag \\
& \qquad \qquad \qquad \qquad  P_f( \mathcal{A}^c \cap \mathcal{B}^c(\eta) \cap T_Q)  \notag \\
& \geq c'_n 2^{-nr} \sum_{Q_X \in \mathcal{Q}_{\mathcal{X}}^n(\alpha,-\delta)}    2^{-n  (R(Q_{X},R+\eta,D,D_e)+ \epsilon) }   \cdot \notag \\
& \qquad \qquad \qquad \qquad P_f( \mathcal{A}^c \cap \mathcal{B}^c(\eta) \cap T_Q)  \label{eqchainconvkey1}
\end{align}
where \begin{enumerate}
\item[(a)] follows from Lemma~\ref{lemmaprobfinal} and~\eqref{eqdefMD}: $m \in \mathcal{M}_D(x^n,k)$ guarantees that $d(x^n,h(m,k)) \leq D$ so that Lemma~\ref{lemmaprobfinal} is applicable. 
\item[(b)] follows from~\eqref{eqdefMDR}: $m \in \mathcal{M}_{D,R} (x^n,k,\eta)$ allows us to restrict the minimum to $Q_{XY} \in \mathcal{Q}_{\mathcal{X} \mathcal{Y}} (Q_X,D)$ satisfying $I_{Q_{XY}} (X;Y) \leq R + \eta$.
 \item[(c)] follows from Proposition~\ref{propinnerminconv} in Appendix~\ref{proplimitproof}.
\end{enumerate} 
Now, note that
\begin{align}
P_f(\mathcal{A} | T_Q) \leq \frac{P_f( \mathcal{A})}{P(T_Q)} \leq (n+1)^{|\mathcal{X}|} 2^{-n( \alpha-D(Q||P))}. \label{eqboundPA} 
\end{align}
Moreover, by Lemma~\ref{LemmaWeissman},
\begin{align}
P_f( \mathcal{B}(\eta) |T_Q) \leq (n+1)^{|\mathcal{X}||\mathcal{Y}|+|\mathcal{X}|} 2^{-n\eta}. \label{eqboundPB} 
\end{align}
Combining~\eqref{eqboundPA} and~\eqref{eqboundPB} yields, for every $Q_X \in \mathcal{Q}_{\mathcal{X}}^n(\alpha,-\delta)$, 
\begin{align}
& P_f( \mathcal{A}^c \cap \mathcal{B}^c(\eta) \cap T_Q) \notag \\
& = P(T_Q) P_f( \mathcal{A}^c \cap \mathcal{B}^c(\eta) |T_Q) \notag \\
& \geq P(T_Q) ( 1 - P_f(\mathcal{A} | T_Q) - P_f(\mathcal{B}(\eta) | T_Q)) \notag \\
& \geq P(T_Q)( 1- (n+1)^{|\mathcal{X}|} 2^{-n\delta} -(n+1)^{|\mathcal{X}||\mathcal{Y}|+|\mathcal{X}|} 2^{-n\eta} ) \notag \\
& \geq P(T_Q) /2,
\end{align}
where the last inequality holds for large enough $n$. Continuing from~\eqref{eqchainconvkey1}, we get
\begin{align}
& \Pr \Big( d_e \Big(X^n,g^\prime(f(X^n,K)) \Big) \leq  D_e \Big) \notag \\
& \geq (c'_n/2) 2^{-nr} (n+1)^{-|\mathcal{X}|} \cdot \notag \\
& \qquad   \sum_{Q_X \in \mathcal{Q}_{\mathcal{X}}^n(\alpha,-\delta)}    2^{-n  (D(Q_X||P)+R(Q_{X},R+\eta,D,D_e)+ \epsilon) }   \notag \\
& \geq (c'_n/2) (n+1)^{-|\mathcal{X}|} \exp \bigg\lbrace -n \bigg( \epsilon+ r +  \notag \\
& \qquad \left. \left. \min_{Q_X \in \mathcal{Q}_{\mathcal{X}}^n(\alpha,-\delta) } D(Q_X||P)+R(Q_{X},R+\eta,D,D_e) \right) \right\rbrace. 
\end{align}
Therefore, taking the limit as $n$ goes to infinity, and noting that $\epsilon$, $\delta$, and $\eta$ are arbitrary, we get
\begin{align}
& E^+(P,\overrightarrow{D},\overrightarrow{R},\alpha)= \notag \\
& \limsup_{n \rightarrow \infty} \max_{\{f_n\}} \min_{\{g_n\}} -\frac{1}{n}\log \Pr \Big( \! d_e \Big(X^n,g_n(f_n(X^n,K)) \Big) \!\! \leq \! D_e \! \Big) \notag \\
& \leq \lim_{\delta \rightarrow 0} \lim_{\eta \rightarrow 0} r+  \!\! \min_{Q: D(Q||P) \leq \alpha-\delta} \!\!  D(Q||P) \! + \! R(Q,R+\eta,D,D_e) \notag \\
& = r + \min_{Q: D(Q||P) \leq \alpha} D(Q||P)+R(Q,R,D,D_e), \label{eqconvkey2}
\end{align}
where the last equality follows similarly to equations~\eqref{eqlimminR} and~\eqref{eqlimmindelta}.
Combining~\eqref{eqconvkey1} and~\eqref{eqconvkey2} yields our result.

\begin{appendices}

\section{Proof of Proposition~\ref{continuity}} \label{continuityproof}

\subsection{Proof of Property (P1)}

(P1): For fixed $P_{XY}$, $R(P_{XY},D_e)$ is a finite valued, non-increasing convex function of $D_e$. Furthermore, $R(P_{XY},D_e)$ is a uniformly continuous function of the pair $(P_{XY},D_e)$.

Fix $P_{XY}$. The minimization in~\eqref{Rsideinfo} is over a compact set, which is non-empty due to assumption (A3). Since $I(X;V|Y)$ is a continuous function of $P_{V|X,Y}$, the minimum is achieved. The monotonicity in $D_e$ follows directly from the definition. It is easy to check that $I(X;V|Y)$ is convex in $P_{V|X,Y}$ for fixed $P_{XY}$. Then, the proof of the convexity of $R(P_{XY},D_e)$ in $D_e$ follows similarly to the case of the rate-distortion function with no side information (see Lemma 2.2.2 in~\cite{korner}).\\

To show the uniform continuity in the pair $(P_{XY},D_e)$, consider the following proposition, the proof of which is given in Appendix~\ref{propcontinuityproof1}.
\begin{Proposition} \label{propcontinuity}
Let $N_1$ and $N_2$ be in $\mathbb{N}$, and let $\mathcal{S}$ and $\mathcal{U}$ be compact subsets of $\mathbb{R}^{N_1}$ and $\mathbb{R}^{N_2}$, respectively. 
Let $\nu$ be a non-negative continuous function defined on $\mathcal{S} \times \mathcal{U}$, and let $\vartheta$ be a real-valued continuous function defined on $\mathcal{S} \times \mathcal{U}$. Suppose they satisfy the following condition:
\begin{enumerate}
\item[(PA)] If  $(s,u_1) \in \mathcal{S} \times \mathcal{U}$ satisfies $\nu(s,u_1)= \min_{u' \in \mathcal{U}} \nu(s,u')$, then there exists $u_2$ such that $\vartheta(s,u_2) = \vartheta(s,u_1)$, and for all $s' \in \mathcal{S}$, $\nu(s',u_2)=\min_{u' \in \mathcal{U}} \nu(s',u').$
\end{enumerate}
Let $t_0 = \max_{s \in \mathcal{S}} \min_{u \in \mathcal{U}} \nu(s,u)$, and let $\varphi$ be a function on $\mathcal{S} \times [t_0, +\infty)$ defined as follows:
\begin{equation*}
\varphi(s,t) = \min_{u: \nu(s,u) \leq t} \vartheta (s,u).
\end{equation*}
If for fixed $s \in \mathcal{S}$, $\varphi(s,t)$ is continuous in $t$, then $\varphi(s,t)$ is continuous in the pair $(s,t)$.
\end{Proposition}  
\bigskip
\begin{Remark}
The proposition generalizes Lemma 2.2.2 in~\cite{korner}, which shows the continuity of the regular rate-distortion function, and the proof follows along similar lines.
\end{Remark}
\bigskip

The proposition yields immediately the continuity of $R(P_{XY},D_e)$ by identifying $\mathcal{S}$ with $\mathcal{P}_{\mathcal{XY}}$, $\mathcal{U}$ with the set of conditional probability distributions $P_{V|XY}$, $t_0$ with $D_{e,\min}$, and the functions $\nu$, $\vartheta$, and $\varphi$ with $\E[d_e(X,V)]$, $I(X;V|Y)$, and $R(P_{XY},D_e)$ respectively. It is easy to check that $D_{e,\min} = \max_{P_{XY}} \min_{P_{V|XY}} \E[d_e(X,V)]$ so that we can identify it with $t_0$. 
To see why  $\E[d_e(X,V)]$ and $I(X;V|Y)$ satisfy (PA), note the following. For notational convenience, we write  $\E[d_e(X,V)]$ as $d_e(P_{XY},P_{V|XY})$, and $I(X;V|Y)$ as $I(P_{XY},P_{V|XY})$. Suppose $d_e(P_{XY},P_{V|XY})= \min_{\hat{P}_{V|XY}} d_e(P_{XY}, \hat{P}_{V|XY})$ and let $D_e(x)=\min_{v \in \mathcal{V}} d_e(x,v)$ for $x \in \mathcal{X}$. Then for all $(x,v)$ such that 
$ d_e(x,v) > D_e(x)$, $P_{XV} (x,v) =0.$ Expanding $P_{XV}(x,v)$: \\ $P_{XV} (x,v)= \sum_{y} P_{XY}(x,y) P_{V|XY}(v|x,y)=0 \Rightarrow$ for all $y \in \mathcal{Y}$, $P_{XY}(x,y)=0$ or $P_{V|XY}(v|x,y)=0$. Then, define $P'_{V|XY}$ as follows:
\begin{itemize}
\item If $P_{XY}(x,y) > 0$, let $P'_{V|(X=x,Y=y)} = P_{V|(X=x,Y=y)}$.
\item If $P_{XY}(x,y) = 0$, let $P'_{V|(X=x,Y=y)}$ satisfy $P'_{V|(X=x,Y=y)} (v |x,y) =0$ if $d_e(x,v) > D_e(x).$
\end{itemize}
Then $P_{XY} P_{V|X,Y} =P_{XY} P'_{V|X,Y}$, thus $I(P_{XY},P_{V|XY}) = I(P_{XY},P'_{V|XY})$. Moreover, the definition of $P'_{V|XY}$ guarantees that $d_e(x,v) > D_e(x) \Rightarrow P'_{V|XY} (v |x,y) =0$ for all $y$. Therefore, for any joint distribution $P'_{XY}$, $d_e(P'_{XY}, P'_{V|XY})= \min_{\hat{P}_{V|XY}} d_e(P_{XY}, \hat{P}_{V|XY})$. 

Finally, to prove uniform continuity, note that $R(P_{XY},D_e)=R(P_{XY},D_{e,\max})$ for all $D_e \geq D_{e,\max}$. Therefore, $R(P_{XY},D_e)$ is uniformly continuous on the set $\mathcal{P}_{\mathcal{XY}} \times [D_{e,\max},\infty)$. Since it is also uniformly continuous on $\mathcal{P}_{\mathcal{XY}} \times [D_{e,\min},D_{e,\max}]$, the result is established. \hfill $\blacksquare$

\subsection{Proof of Property (P2)}
(P2): For fixed $P_X$, $R(P_X,D,D_e)$ is a finite-valued function of $(D,D_e)$. Moreover, for fixed $D_e$, $R(P_X,D,D_e)$ is a uniformly continuous function of the pair $(P_X,D)$.

Fix $P_X$. The maximization in~\eqref{Rmaxmin} is over a compact set, which is non-empty due to assumption (A3). Since $R(P_{XY},D_e)$ is a continuous function of $P_{XY}$, it is also continuous in $P_{Y|X}$ for fixed $P_X$. Therefore, the maximum is achieved. 

As for the continuity in $(P_X,D)$ for fixed $D_e$, we view $R(P_X,D,D_e)$ as a function of $(P_X,D)$, and $R(P_{XY},D_e)$ as function of $(P_X,P_{Y|X})$. In the terminology of Proposition~\ref{propcontinuity}, we identify $\mathcal{S}$ with $\mathcal{P}_X$, $~\mathcal{U}$ with the set of conditional probability distributions $P_{Y|X}$, $t_0$ with $D_{\min}$, and the functions $\nu$, $\vartheta$, and $\varphi$ with $\E[d(X,Y)]$, $-R(P_{XY},D_e)$, and $-R(P_X,D,D_e)$ respectively. Proving that $\E[d(X,Y)]$ and $R(P_{XY},D_e)$ satisfy (PA) follows along the same lines as proving $\E[d_e(X,V)]$ and $I(X;V|Y)$ satisfy (PA) . Moreover, if continuity holds, uniform continuity follows from the fact that $R(P_X,D,D_e)$ is constant for all $D \geq D_{\max}$.

 It remains to show that $R(P_X,D,D_e)$ is a continuous function of $D$ for fixed $P_X$ and $D_e$. The result of Proposition~\ref{propcontinuity} then applies immediately. To this end, consider the following proposition, the proof of which is given in Appendix~\ref{propcontinuityproof2} 
\begin{Proposition} \label{propcontinuityD}
Let $N$ be in $\mathbb{N}$, and let $\mathcal{T}$ be a non-empty compact subset of $\mathbb{R}^{N}$. Let $L$ be a real-valued continuous function defined on $\mathcal{T}$. Let $T_1 \supseteq T_2 \supseteq \cdots$ be a decreasing sequence of non-empty compact subsets of $\mathcal{T}$. Let $T= \bigcap_{i \geq 1} T_i$. Then,
\begin{equation*}
\lim_{k \rightarrow \infty} \max_{t \in T_k} L(t) = \max_{t \in T} L(t).       
\end{equation*} 
Moreover, let $S_1 \subseteq S_2 \subseteq \cdots$ be an increasing sequence of non-empty compact subsets of $\mathcal{T}$. Let $S = \overline{\bigcup_{i \geq 1} S_i}$ (where the bar denotes closure of the set). Then
\begin{equation*}
\lim_{k \rightarrow \infty} \max_{t \in S_k} L(t) = \max_{t \in S} L(t).        
\end{equation*} 
Consequently, if $\mathcal{T}$ is also convex, and $L_c$ is a real-valued convex and continuous function defined on $\mathcal{T}$ with $s_0 = \min_{t \in \mathcal{T}} L_c(t)$, then
\begin{align*}
\hat{L}(s) := \max_{t: L_c(t) \leq s} L(t)
\end{align*}
is continuous in $s \in [s_0, +\infty)$.
\end{Proposition}

It follows immediately then that $R(P_X,D,D_e)$ is continuous in $D$ for fixed $P_X$ and $D_e$, since $\E[d(X,Y)]$ is convex and continuous in $P_{Y|X}$, and $R(P_{XY},D_e)$ is continuous in $P_{Y|X}$ (for fixed $P_X$).

\subsection{Proof of Property (P3)}
(P3): $R_e(P_X,D_e)-R(P_X,D) \leq R(P_X,D,D_e) \leq R_e(P_X,D_e).$

The upper bound is straightforward since $R(P_{XY},D_e)$, the rate-distortion function with side information, is always upper-bounded by $R_e(P_X,D_e)$. The lower bound is derived by considering a conditional $P_{Y|X}^\star$ that achieves the rate-distortion function.
\begin{align}
& R(P_X,D,D_e) \notag \\
 &  = \max_{ \substack{ P_{Y|X}: \\ \E[d(X,Y)] \leq D}} \min_{ \substack{P_{V|X,Y}:\\ ~~\E[d_e(X,V)] \leq D_e}} I(X;V|Y) \notag \\
& \geq \!\! \min_{ \substack{P_{V|X,Y}:\\ ~~\E[d_e(X,V)] \leq D_e}} \!\! \!\! \!\! H_{P_{XY}^\star}(X|Y) - H_{P_{XY}^\star P_{V|XY}} (X|V,Y) \label{eqp3proof1} \\
& \geq \!\! \min_{ \substack{P_{V|X,Y}:\\ ~~\E[d_e(X,V)] \leq D_e}} \hspace{-5mm}H_{P_{XY}^\star}(X|Y) - H_{P_{X} P_{V|X}} (X|V) \label{eqp3proof2} \\
& = -H_{P_X}(X)+H_{P_{XY}^\star}(X|Y) + \notag \\
& \qquad \qquad  \min_{ \substack{P_{V|X,Y}:\\ ~~\E[d_e(X,V)] \leq D_e}} \hspace{-5mm} H_{P_X}(X) - H_{P_{X} P_{V|X}} (X|V) \notag \\
& =  -R(P_X,D)+R_e(P_X,D_e). \notag
\end{align}

%\section{Proof of Lemma~\ref{lemmacode}}
%\label{lemmacodeproof}
%
%\input{AppendixCoveringLemma.tex}
%

\section{Proof of Proposition~\ref{proplimit}} 
\label{proplimitproof}

First, consider the following proposition.
\begin{Proposition} \label{propinnerminconv}
For all $\epsilon > 0$, there exists $n_2(\epsilon,|\mathcal{X}|,|\mathcal{Y}|,|\mathcal{V}|)$, such that for all $n \geq n_2$, for all $D_e \geq D_{e,\min}$, for each $Q_{XY} \in \mathcal{Q}_{XY}^n$,
\begin{equation*}
\Big | \min_{ \substack{ P_{XYV} \in \\ \mathcal{Q}_{XYV}^n (Q_{XY},D_e)}} I_{P_{XYV}} (X;V|Y) - R(Q_{XY},D_e) \Big| \leq \epsilon.
\end{equation*}
\end{Proposition}
\begin{proof}
It follows directly from the definition that 
\begin{equation*}
\min_{ \substack{ P_{XYV} \in \\ \mathcal{Q}^n (Q_{XY},D_e)}} I_{P_{XYV}} (X;V|Y) \geq R(Q_{XY},D_e).
\end{equation*}
So, we only need to show the other direction. To that end, let $\delta > 0$ be small enough such that 
\begin{align} 
\| P_{XYV} - P'_{XYV} \| \leq \delta \Rightarrow \notag \\
 \left| I_{P_{XYV}} (X;V|Y) - I_{P'_{XYV}} (X;V|Y) \right| \leq \epsilon, \label{propinnereq3}
\end{align}
where $\| \cdot \|$ is used to indicate the $L_2$-norm. 
Let $n \geq n_2 \geq |\mathcal{V}|\sqrt{|\mathcal{X}||\mathcal{Y}||\mathcal{V}|} / \delta$. Fix $Q_{XY} \in \mathcal{Q}_\mathcal{XY}^n$, and let $P^\star_{V|XY}$ be the conditional distribution achieving the minimum in $R(Q_{XY},D_e)$. We construct a conditional distribution $P'_{V|XY}$ as follows. For each $(x,y) \in \mathcal{X} \times \mathcal{Y}$, we will choose $P'_{V|X=x,Y=y}$ from $\mathcal{Q}_\mathcal{V}^{nQ_{XY}(x,y)}$, i.e., the set of rational PMFs over $\mathcal{V}$ with denominator $nQ_{XY}(x,y)$ (if $Q_{XY}(x,y)=0$, then we can choose $P'_{V|X=x,Y=y}$ to be any distribution). This guarantees that $Q_{XY}P'_{V|XY}$ is in $\mathcal{Q}_{\mathcal{XYV}}^n $. Let $v(x) = \argmin_{v \in \mathcal{V}} d_e(x,v)$ for $x \in \mathcal{X}$ (if more than one $v$ achieves the minimum, choose one arbitrarily). We construct $P'_{V|XY}$ by rounding $P^\star_{V|XY}$ as follows. For each $(x,y) \in \mathcal{X} \times \mathcal{V}$, for $v \neq v(x)$, we set $P'_{V|XY}(v|x,y)$ to be the largest integer multiple of $1/(nQ_{XY}(x,y))$ that is smaller than $P^\star_{V|XY}(v|x,y)$, i.e., we round down with resolution $1/(nQ_{XY}(x,y))$ and denote this operation by $\lfloor . \rfloor_{nQ_{XY}(x,y)}$. Finally, we set $P'_{V|XY}(v(x)|x,y)$ appropriately to make $P'_{V|XY}(.|x,y)$ a valid probability distribution. It is easy to see that, for such a choice, 
\begin{equation*}
\left|P'_{V|XY}(v|x,y)-P^\star_{V|XY} (v|x,y) \right| \leq \frac{|\mathcal{V}|}{nQ_{XY}(x,y)}.
\end{equation*}
Moreover, this readily implies that
\begin{equation} \label{propinnereq4}
\left\| Q_{XY} P'_{V|XY} - Q_{XY}P^\star_{V|XY} \right\| \leq \frac{|\mathcal{V}|\sqrt{|\mathcal{X}||\mathcal{Y}||\mathcal{V}|}}{n} \leq \delta.
\end{equation}
Let $P^\star_{XYV} = Q_{XY} P^\star_{V|XY}$, and $P'_{XYV} = Q_{XY} P'_{V|XY}$. Now, note that
\begin{align*}
D_e & \geq \E_{P^\star_{XYV}} [d_e(X,V)]  \\
& = \sum_{x,y} \sum_{v} Q_{XY}(x,y) P^\star_{V|XY}(v|x,y) d_e(x,v)  \\
&  =  \sum_{x,y} \sum_{v} Q_{XY}(x,y) \lfloor P^\star_{V|XY}(v|x,y)\rfloor_{nQ_{XY}(x,y)} d_e(x,v)  \\
& \quad +  \sum_{x,y} \sum_{v} Q_{XY}(x,y) \cdot \\
&  \left( P^\star_{V|XY}(v|x,y) - \lfloor P^\star_{V|XY}(v|x,y)\rfloor_{nQ_{XY}(x,y)} \right) d_e(x,v) \\
& \geq \sum_{x,y} \sum_{v} Q_{XY}(x,y) \lfloor P^\star_{V|XY}(v|x,y)\rfloor_{nQ_{XY}(x,y)} d_e(x,v)  \\
& \quad +  \sum_{x,y} \sum_{v} Q_{XY}(x,y) \cdot \\
& \left( P^\star_{V|XY}(v|x,y) \! - \lfloor P^\star_{V|XY}(v|x,y)\rfloor_{nQ_{XY}(x,y)} \! \right) \! d_e(x,v(x)) \\
& = \E_{P'_{XYV}} [d_e(X,V)].
\end{align*}
Therefore, 
\begin{align*}
\min_{\substack{P_{XYV} \in \\ \mathcal{Q}^n(Q_{XY},D_e)}} I_{P_{XYV}}(X;V|Y) & \leq I_{P'_{XYV}}(X;V|Y) \\
& \leq  I_{P^\star_{XYV}}(X;V|Y) + \epsilon \\
& = R(Q_{XY},D_e) + \epsilon,
\end{align*}
where the second inequality follows from~\eqref{propinnereq3} and~\eqref{propinnereq4}.

\end{proof}

Similarly, we have the following proposition.
\begin{Proposition} \label{propoutermaxconv}
For all $\epsilon > 0$, there exists $n_3(\epsilon,|\mathcal{X}|,|\mathcal{Y}|,d_e)$, such that for all $n \geq n_3$, $D \geq D_{\min}$, $D_e \geq D_{e,\min}$, and for each $Q_{X} \in \mathcal{Q}_{X}^n$, 
\begin{equation*}
\Big | \max_{ \substack{ P_{XY} \in \\ \mathcal{Q}_{XY}^n (Q_{X},D)}} R(Q_{XY},D_e) - R(Q_X,D,D_e) \Big| \leq \epsilon.
\end{equation*}
\end{Proposition}
The proof follows along the same lines as that of Proposition~\ref{propinnerminconv}, and is thus omitted. \hfill $\blacksquare$ \\

By the previous two propositions, for any given $\epsilon > 0$, we can set $n$ large enough to satisfy
\begin{equation*}
 \left | I_{P^\star_n(Q_X)} (X;V|Y) - R(Q_X,D,D_e) \right| \leq \epsilon,
\end{equation*} 
for all $Q_X \in \mathcal{Q}_X^n$. Therefore,
\begin{align*}
\min_{Q_X \in \mathcal{Q}_X^n} D(Q_X||P)+R(Q_X,D,D_e) - \epsilon \leq \\
\min_{Q_X \in \mathcal{Q}_X^n} D(Q_X||P)+ I_{P^\star_n(Q_X)} (X;V|Y) \leq \\
\min_{Q_X \in \mathcal{Q}_X^n} D(Q_X||P)+R(Q_X,D,D_e) + \epsilon
\end{align*}
By taking the limit as $n$ goes to infinity, and noting that $\epsilon$ is arbitrary, the proof is concluded.

\section{Proof of Proposition~\ref{propcontinuityR}}
\label{AppendixpropcontinuityR}

\subsection{Proof of Property (P4)}

(P4): For fixed $P_X$, $R(P_X,R,D,D_e)$ is a finite-valued function of $(R,D,D_e)$. Moreover, for fixed $D_e$,
 $R(P_X,R,D,D_e)$ is continuous in the triple $(P_X,R,D)$ over the set $ \mathcal{S} = \left\lbrace (P_X,R,D) : P_X \in \mathcal{P}_{\mathcal{X}}, D \geq D_{\min},R > R(P_X,D) \right\rbrace $.

Recall, for $P_X$ satisfying $R(P_X,D) \leq R$,
\begin{align*}
R(P_X,R,D,D_e) = \max_{\substack{ P_{Y|X}: \\ \E[d(X,Y)] \leq D \\ I(X;Y) \leq R}} R(P_{XY},D_e).
\end{align*}
For fixed $P_X$, let $S_{D,R} = \left\lbrace P_{Y|X}: \E[d(X,Y)] \leq D, I(X;Y) \leq R \right\rbrace$. Then $S_{D,R}$ is compact, and non-empty since $D \geq D_{\min}$ and $R \geq R(P_X,D)$. Since $R(P_{XY},D_e)$ is a continuous function of $P_{XY}$ (by Proposition~\ref{continuity}), it is also continuous in $P_{Y|X}$ for fixed $P_X$. Therefore, the maximum is achieved. 

To prove continuity of $R(P_X,R,D,D_e)$ in $(P_X,R,D)$, first consider the following claims.

{\emph{Claim 1:}} For fixed $P_X$, $D_e$, and $D$,  $R(P_X,R,D,D_e)$ is continuous in $R$, where $R \in [R(P_X,D), +\infty)$.

This follows from the third part of Proposition~\ref{propcontinuityD} in Appendix~\ref{continuityproof} by identifying $\mathcal{T}$ with $\{P_{Y|X}: \E[d(X,Y)] \leq D \}$ (which is compact, convex, and non-empty since $D \geq D_{\min}$), $L_c$ with $I(X;Y)$ which is convex and continuous in $P_{Y|X}$, $s_0$ with $R(P_X,D)$, $L$ with $R(P_{X}P_{Y|X},D_e)$, and $\hat{L}$ with $R(P_X,R,D,D_e)$.

{\emph{Claim 2:}} For fixed $P_X$, $D_e$, and $R$,  $R(P_X,R,D,D_e)$ is continuous in $D$, where $D \in [D(P_X,R), +\infty)$ and $D(P_X,R) := \min_{P_{Y|X}: I(X;Y) \leq R} \E[d(X,Y)]$ is the distortion-rate function.

This follows from a similar argument.
%\begin{IEEEproof} 
%Note that $S_R = \{ P_{Y|X}: I(X;Y) \leq R \}$ is a convex set. Consequently, $\E[d(X,Y)]$ defined over $S_R$ is a convex function. Given this observation, the proof follows along the same lines as that of Claim 1, with the roles of $I(X;Y)$ and $\E[d(X,Y)]$ flipped in this case.
%\end{IEEEproof}
%\bigskip

We are now ready to prove continuity in the triple $(P_X,R,D)$ over $\mathcal{S} = \left\lbrace (P_X,R,D) : P_X \in \mathcal{P}_{\mathcal{X}}, D \geq D_{\min},R > R(P_X,D) \right\rbrace$. 

To that end, fix any $(P,R,D) \in \mathcal{S}$ and consider any sequence $(P_k,R_k,D_k)$ converging to  $(P,R,D)$. First, we show that $ \liminf_{k \rightarrow \infty} R(P_k,R_k,D_k,D_e) \geq R(P,R,D,D_e).$ 
 Consider any $\epsilon > 0$. By continuity of $R(P,R,D,D_e)$ in $R$ (for fixed $P$, $D$, and $D_e$), we can choose $R'$ such that $R(P_X,D) < R' < R$ and $ R(P,R',D,D_e) \geq R(P,R,D,D_e) - \epsilon/2$.  We now consider two cases depending on the value of $D$. Let $D_0 = \min_{P_{Y|X}} \E[d(X,Y)]$.
 
 If $D > D_0$: note that $D(P,R)$ is non-increasing in $R$, therefore $D'(P,R) \leq 0$. Moreover, it is convex in $R$ and $R(P,D)$ does not achieve its minimum ($D(P,R(P,D))=D>D_0$), hence $D'(P,R(P,D)) < 0$. Therefore, $R > R(P,D) \Rightarrow D(P,R) < D$. Now 
 choose $D'$ such that $D(P,R) < D' < D$ and $R(P,R',D',D_e) \geq R(P,R',D,D_e) - \epsilon/2$. Let $P_{Y|X}^\star$ be a maximizer for $R(P,R',D',D_e)$. 
 
 If $D=  D_0$: set $D'=D$ and $P'_{Y|X}$ be a maximizer for $R(P,R',D,D_e)$. 
 Let $D(x) = \min_{y \in \mathcal{Y}} d(x,y)$ for $x \in \mathcal{X}$. Then $P'_{Y|X}$ must satisfy the following property: for all $(x,y)$ such that $d(x,y) > D(x)$, $P(x) = 0$ or $P'_{Y|X} (y|x) = 0$. We can construct $P_{Y|X}^\star$ such that $d(x,y) > D(x) \Rightarrow P_{Y|X}^\star(y|x) = 0$, and $P(x) > 0 \Rightarrow P^\star_{Y|X=x} = P'_{Y|X=x}$. As such, $PP'_{Y|X} = PP_{Y|X}^\star$.
 
 We claim that $P^\star_{Y|X}$ is feasible for the maximization in $R(P_k,R_k,D_k,D_e)$ for sufficiently large $k$. Indeed, $I(P_k;P^\star_{Y|X}) \rightarrow I(P;P^\star_{Y|X}) \leq R' < R$. Then for sufficiently large $k$, $I(P;P^\star_{Y|X}) \leq R_k$. Moreover, if $ D > D_0$, then $\E[d(P_k,P^\star_{Y|X})] \rightarrow \E[d(P,P^\star_{Y|X})] \leq D' < D$. Then for sufficiently large $k$, $\E[d(P_k,P^\star_{Y|X})] \leq D_k$. Similarly, if $D=D_0$, then $\E[d(P_k,P^\star_{Y|X})] = \min_{P_{Y|X}} \E[d(P_k,P_{Y|X})] \leq D_{\min} \leq D_k$, where the first equality follows from the construction of $P^\star_{Y|X}$. So we get
 \begin{align*}
 \liminf_{k \rightarrow \infty} R(P_k,R_k,D_k,D_e) & \geq \liminf_{k \rightarrow \infty} R(P_k P^\star_{Y|X},D_e) 	\\
 & = R( P P^\star_{Y|X},D_e) \\
 & = R(P,R',D',D_e) \\
 & \geq R(P,R,D,D_e) -\epsilon,
 \end{align*}
 where the first equality follows from the continuity of $R(P_{XY},D_e)$ in $P_{XY}$. Noting that $\epsilon$ is arbitrary, we get our first inequality.
 
 On the other hand, let $P_{Y|X}^{(k)}$ be a maximizer for $R(P_k,R_k,D_k,D_e)$. Consider a sequence of integers $\{k_j\}$ such that 
 \begin{align*}
 R(P_{k_j},R_{k_j},D_{k_j},D_e) \rightarrow \limsup_{k \rightarrow \infty}  R(P_k,R_k,D_k,D_e).
 \end{align*}
Let $P_{Y|X}^{(k_j)}$ be the corresponding subsequence of maximizers. Since the set of conditional distributions $\{P_{Y|X}\}$ is bounded, $\{ P_{Y|X}^{(k_j)} \}$ has a convergent subsequence $P_{Y|X}^{(k_{j_\ell})}$. Let $P^\star_{Y|X}$ be its limit. We have, $I(P;P^\star_{Y|X}) = \lim_{\ell \rightarrow \infty} I(P;P^{(k_{j_{\ell}})}_{Y|X}) \leq \lim_{\ell \rightarrow \infty} R_{k_{j_\ell}} = R$. Similarly, $\E[d(P,P^\star_{Y|X})] = \lim_{\ell \rightarrow \infty} \E[d(P,P^{(k_{j_{\ell}})}_{Y|X})] \leq \lim_{\ell \rightarrow \infty} D_{k_{j_\ell}} = D$. Therefore,
\begin{align*}
R(P,R,D,D_e) & \geq R(PP^\star_{Y|X},D_e) \\
& = \lim_{\ell \rightarrow \infty } R(P_{k_{j_\ell}} P_{Y|X}^{k_{j_\ell}}, D_e) \\
& = \lim_{\ell \rightarrow \infty} R(P_{k_{j_\ell}},R_{k_{j_\ell}},D_{k_{j_\ell}},D_e) \\ 
&= \limsup_{k \rightarrow \infty} R(P_k,R_k,D_k,D_e).
\end{align*}

\subsection{Proof of Property (P5)}

(P5): $R_e(P_X,D_e) - R(P_X,D) \leq R(P_X,R,D,D_e) \leq R(P_X,D,D_e) \leq R_e(P_X,D_e).$

The upper bound follows straightforwardly from the definition and (P3). The lower bound follows from the proof of (P3). Indeed, the bound in (P3) was derived by considering a conditional $P^\star_{Y|X}$ that achieves the rate-distortion function. As such, this choice is feasible since $I_{P^\star_{XY}} (X;Y) = R(P_X,D) \leq R$. 
\section{Proof of Lemma~\ref{lemmahamming}}
\label{lemmahammingproof}
Note that the second equality follows simply from the evaluation of $R_e(Q,D_e)-R(Q,D)$. So we only need to show the first equality.

Note that (P3) asserts that $R(Q,D,D_e) \geq R_e(Q,D_e) - R(Q,D)$, so we only need to show the reverse direction. Moreover, 
if $H(Q) \leq H(D)$, $R(Q,D)=0$. It then follows from (P3) that $R(Q,D,D_e)=R_e(Q,D_e)$. It remains to show that, for $Q$ satisfying $H(Q) \geq H(D)$,
\begin{align*}
R(Q,D,D_e) \leq R_e(Q,D_e) - R(Q,D).
\end{align*}
\begin{Remark}
The following proof was suggested by the reviewer, and it significantly simplifies our previous proof.
\end{Remark}
To that end, let $P_{Y|X}$ satisfy $\E[d(X,Y)] \leq D$, $\tilde{X}=X \oplus Y$, $\tilde{V} = V \oplus Y$ and consider
\begin{align*}
& \min_{P_{V|XY}: \E[d(X,V)] \leq D_e} I(X;V|Y) \\
& =  \min_{P_{V|XY}: \E[d(X\oplus Y,V\oplus Y)] \leq D_e} I(X\oplus Y;V\oplus Y|Y) \\
&  = \min_{P_{\tilde{V}|\tilde{X} Y}: \E[d( \tilde{X}, \tilde{V}) ] \leq D_e} I(\tilde{X};\tilde{V}|Y) \\
& \leq \min_{\substack{ P_{  \tilde{V}|\tilde{X} }: \E[d( \tilde{X}, \tilde{V}) ] \leq D_e \\ V-X-Y } } I(\tilde{X};\tilde{V}|Y) \\
& \leq \min_{P_{\tilde{V}|\tilde{X} }: \E[d( \tilde{X}, \tilde{V}) ] \leq D_e} I(\tilde{X};\tilde{V}) \\
& = [H(\tilde{X}) - H(D_e)]^+ 
\stackrel{\text{(a)}} \leq H(D)-H(D_e),
\end{align*}
where (a) follows from the fact that $\Pr( \tilde{X} = 1) = \E[d(X,Y)] \leq D.$ Therefore,
\begin{align*}
R(Q,D,D_e) & = \max_{\substack{P_{Y|X}:  \E[d(X,Y)] \leq D}} \min_{ \substack{P_{V|XY}: \\ \E[d(X,V)] \leq D_e}} I(X;V|Y) \\
& \leq H(D)-H(D_e) = R_e(Q,D_e)-R(Q,D),
\end{align*}
as desired. \hfill $\blacksquare$

\section{Proof of Lemma~\ref{lemmaachievability}} \label{mainlemmaproof}

{\bf Proof of 1):}
For $x^n \in T_{Q_X}$, and $m \in [N]$, let
\begin{align*}
N_{x^n,m} = \mathbb{I} \{(x^n,Y_m^n) \in T_{Q_{XY}} \}, \text{ so that } N_{x^n} = \!\sum_{m=1}^N \! N_{x^n,m}. 
\end{align*}
Note that, $N_{x^n,m} \sim Ber(\beta)$, where
\begin{align}
\beta = \Pr ( (x^n,Y_m^n) \in T_{Q_{XY}}) = \frac{|T_{Q_{Y|X}}(x^n)|}{|T_{Q_Y}|}, \notag \\
\Rightarrow 2^{-n(I_{Q_{XY}} (X;Y) + \epsilon/2 )} \leq \beta \leq  2^{-n(I_{Q_{XY}} (X;Y) -\epsilon/2 )}. \label{eqBernbound}
\end{align}
Therefore,
\begin{align}
\Pr(N_{x^n} = 0) & = \Pr (N_{x^n,m}=0, \forall m \in [N]) \notag \\
& \stackrel{\text{(a)}}= \prod_{m=1}^N (1-\beta) \stackrel{\text{(b)}} \leq e^{-\beta N}  \leq e^{-2^{n \epsilon/6}}, \label{eqNxnzero}
\end{align}
where (a) follows from the independence of $N_{x^n,m}$ for different $m$'s, and (b) follows from the fact that $(1-t)^N \leq e^{-tN}$. On the other hand, 
\begin{align}
\Pr(N_{x^n} > 2^{2n\epsilon}) & = \Pr \left( \sum_{m=1}^N N_{x^n,m} > 2^{2n\epsilon} \right) \notag \\
& \stackrel{\text{(a)}} \leq \left( \frac{eN\beta}{2^{2n\epsilon}} \right)^{2^{2n\epsilon}} \leq (e2^{-n \epsilon/2})^{2^{2n\epsilon}}, \label{eqNxntoolarge}
\end{align}
where (a) follows from the Chernoff bound (cf.~\cite[Lemma 2]{cuff2014henchman}). Using equations~\eqref{eqNxnzero} and~\eqref{eqNxntoolarge} and the union bound, we get
\begin{align*}
\Pr(\mathcal{E}) \leq |\mathcal{X}|^n \left( (e2^{-n \epsilon/2})^{2^{2n\epsilon}} + e^{-2^{n \epsilon/6}} \right) \leq e^{-2^{n \epsilon/7}},
\end{align*} establishing~\eqref{eqlemmaach1}. \hfill $\blacksquare$

{\bf Proof of 2):} To show that~\eqref{eqlemmaach2} holds, consider $\mathcal{C}^n \notin \mathcal{E}$, and $(x^n,m)$ where $m \in C(x^n)$, \begin{align}
P_{X^n|M}^{\mathcal{C}} (x^n|m) & = \frac{P_{M|X^n}^{\mathcal{C}} (m|x^n)}{\sum_{x^n \in T_{Q_{X|Y}}(y_m^n)} P_{M|X^n}^{\mathcal{C}} (m|x^n) } \notag \\
&  \leq \frac{1}{2^{n(H_{Q_{XY}}(X|Y)-\epsilon)}2^{-2n\epsilon} } \notag \\
& = 2^{-n(H_{Q_{XY}}(X|Y)-3\epsilon)}. \label{eqboundpx|y}
\end{align}
Then,
\begin{align}
\label{eqboundguessing}
& \Pr(d_e(X^n,v^n) \leq D_e |M=m,\mathcal{C}^n) \notag \\
 & = \sum_{x^n: d_e(x^n,v^n) \leq D_e} P_{X^n|M}^{\mathcal{C}} (x^n|m) \notag \\
& \leq \sum_{ \substack{ x^n: d_e(x^n,v^n) \leq D_e \\ (x^n,y_m^n) \in T_{Q_{XY}}}}  2^{-n(H_{Q_{XY}}(X|Y)-3\epsilon)} \notag \\
& \stackrel{\text{(a)}} \leq  \max_{P_{XYV} \in \mathcal{Q}^n(Q_{XY},D_e)} \!\!\!\! 2^{ n \left(H_{P_{XYV}} ( X|V,Y)+\epsilon \right) } 2^{-n(H_{Q_{XY}}(X|Y)-3\epsilon)} \notag \\
& =  \max_{P_{XYV} \in \mathcal{Q}^n(Q_{XY},D_e)} 2^{-n(I_{P_{XYV}} (X;V|Y) - 4\epsilon)} \notag \\
&  \leq 2^{-n(R(Q_{XY},D_e)-4\epsilon)},
\end{align}
where (a) follows from~\eqref{eavesbound}. It remains to show~\eqref{eqlemmaach3}. To that end, note that, given $y^n \in \mathcal{Y}^n$ and $m \in [N]$,
\begin{align*}
\Pr(Y_m^n = y^n | \mathcal{E}^c )  & \leq \frac{\Pr(Y_m^n = y^n)}{\Pr(\mathcal{E}^c)} \\
& \leq \frac{2^{-n(H_{Q_Y}(Y) - \epsilon/2)}}{1-e^{-2^{n\epsilon/7}}} \leq 2^{-n(H_{Q_Y}(Y) - \epsilon)} .
\end{align*}
Therefore,
\begin{align*} 
& \E [ \Pr(d_e(X^n,v^n) \leq D_e |M=m,\mathcal{C}^n) | \mathcal{E}^c ] \\
& =  \sum_{\mathcal{C}^n \in \mathcal{E}^c} \Pr(\mathcal{C}^n | \mathcal{E}^c) \Pr(d_e(X^n,v^n) \leq D_e |M=m,\mathcal{C}^n) \\
& =     \sum_{y^n \in \mathcal{Y}^n} \sum_{\mathcal{C}^n \in \mathcal{E}^c} \Pr(Y_m^n=y^n | \mathcal{E}^c) \Pr(\mathcal{C}^n| Y_m^n=y^n , \mathcal{E}^c)  \cdot \\
& \qquad \qquad \qquad \quad \Pr(d_e(X^n,v^n) \leq D_e |M=m,\mathcal{C}^n)  \\
& =   \sum_{y^n \in \mathcal{Y}^n} \sum_{\mathcal{C}^n \in \mathcal{E}^c} \Pr(Y_m^n=y^n | \mathcal{E}^c) \Pr(\mathcal{C}^n| Y_m^n=y^n , \mathcal{E}^c)  \cdot \\
& \qquad \qquad \qquad  \sum_{ \substack{ x^n: d_e(x^n,v^n) \leq D_e \\ (x^n,y^n) \in T_{Q_{XY}}}}  P_{X^n|M}^{\mathcal{C}} (x^n|m) \\
& \stackrel{\text{(a)}} \leq  \sum_{y^n \in \mathcal{Y}^n} \sum_{\mathcal{C}^n \in \mathcal{E}^c} \Pr(Y_m^n=y^n | \mathcal{E}^c) \Pr(\mathcal{C}^n| Y_m^n=y^n , \mathcal{E}^c) \cdot \\
& \qquad \qquad  \sum_{ \substack{ x^n: d_e(x^n,v^n) \leq D_e \\ (x^n,y^n) \in T_{Q_{XY}}}}  2^{-n(H_{Q_{XY}}(X|Y)-3\epsilon)} \\
& =   \sum_{y^n \in \mathcal{Y}^n} \Pr(Y_m^n=y^n | \mathcal{E}^c) \!\!\!\! \sum_{ \substack{ x^n: d_e(x^n,v^n) \leq D_e \\ (x^n,y^n) \in T_{Q_{XY}}}} \!\!\!\!\!\!  2^{-n(H_{Q_{XY}}(X|Y)-3\epsilon)} \\
& =  \sum_{x^n: d_e(x^n,v^n) \leq D_e} \sum_{y^n \in T_{Q_{Y|X}(x^n)}} \Pr(Y_m^n=y^n | \mathcal{E}^c) \cdot \\
& \qquad \qquad \qquad \qquad \qquad \qquad \quad ~ 2^{-n(H_{Q_{XY}}(X|Y)-3\epsilon)} \\
& \leq  \sum_{x^n: d_e(x^n,v^n) \leq D_e} \sum_{y^n \in T_{Q_{Y|X}(x^n)}} 2^{-n(H_{Q_Y}(Y) - \epsilon/2)} \cdot \\
& \qquad \qquad \qquad \qquad \qquad \qquad \quad ~~ 2^{-n(H_{Q_{XY}}(X|Y)-3\epsilon)} \\
& \leq  \sum_{x^n: d_e(x^n,v^n) \leq D_e} 2^{nH_{Q_{XY}} (Y|X)} 2^{-n(H_{Q_{XY}} (X,Y) - 7\epsilon/2)} \\
& =  \sum_{x^n: d_e(x^n,v^n) \leq D_e}  2^{-n(H_{Q_{X}} (X) - 7\epsilon/2)} \\
& \stackrel{\text{(b)}} \leq \max_{ \substack{ P_{XV} \in \mathcal{Q}_{\mathcal{X} \mathcal{V}}^n}} 2^{n(H_{P_{XV}} (X|V) + \epsilon/2)} 2^{-n(H_{Q_{X}} (X) - 7\epsilon/2)}  \\
& \leq 2^{-n(R_e(Q_X,D_e)-4\epsilon)},
\end{align*}
where (a) follows from~\eqref{eqboundpx|y}, and (b) can be shown analogously to~\eqref{eavesbound}. \hfill $\blacksquare$

{\bf Proof of 3):}  For notational convenience, let $E= \min\{ R_e(Q_X,D_e), r+R(Q_{XY},D_e)\}$. Note that,
\begin{align} 
 \Pr (\tilde{\mathcal{E}}) 
  & \leq \Pr \left( \bigcup_{k=1}^{2^{nr}} \mathcal{E}_k \right)  
   + \Pr \left(  \tilde{\mathcal{E}} \bigg| \left( \bigcup_{k=1}^{2^{nr}} \mathcal{E}_k \right)^c \right) \notag \\
&    \leq e^{-2^{n\epsilon/8}} 
   + \Pr \left( \tilde{\mathcal{E}} \bigg|  \bigcap_{k=1}^{2^{nr}} \mathcal{E}_k^c \right), \label{eqachfinalboundinit}
\end{align}
where the second inequality follows from the union bound and~\eqref{eqlemmaach1}. Now, 
fix $ \{\mathcal{C}_k^n\}_{k=1}^{2^{nr}} \in \cap_{k=1}^{2^{nr}} \mathcal{E}_k^c$, $m \in [N]$, and $v^n \in \mathcal{V}^n$, and suppose  $K=k_0$. Then,
\begin{align}
& \Pr(d_e(X^n,v^n) \leq D_e |M=m,K=k_0, \{\mathcal{C}_k^n\}_{k=1}^{2^{nr}}) \notag \\
& = \Pr(d_e(X^n,v^n) \leq D_e | M=m, \mathcal{C}_{k_0}^n) \notag \\
& \leq 2^{-n(R(Q_{XY},D_e) - 4\epsilon)}, \label{eqachfinalboundprob}
\end{align}
where the inequality follows from~\eqref{eqlemmaach2}. Furthermore, 
\begin{align}
 & \E \!\! \left[ \! \Pr(d_e(X^n,v^n) \! \leq \! D_e |M \! =m,K \! =k_0, \{\mathcal{C}_k^n\}_{k=1}^{2^{nr}}) \Bigg| \bigcap_{k=1}^{2^{nr}} \mathcal{E}_k^c \right] \notag \\
& = \E \left[ \Pr(d_e(X^n,v^n) \leq D_e | M=m, \mathcal{C}_{k_0}^n) \Bigg| \bigcap_{k=1}^{2^{nr}} \mathcal{E}_k^c \right] \notag \\
& =  \E \left[ \Pr(d_e(X^n,v^n) \leq D_e | M=m, \mathcal{C}_{k_0}^n)|  \mathcal{E}_{k_0}^c \right] \notag \\
& \leq 2^{-n(R_e(Q_X,D_e)-4 \epsilon)}, \label{eqachfinalboundexp}
\end{align} 
where the last inequality follows from~\eqref{eqlemmaach3}. Now, consider $\{\mathcal{C}_k^n\}_{k=1}^{2^{nr}} \in \left(\cup_{k=1}^{2^{nr}} \mathcal{E}_k \right)^c$.
\begin{align}
& \Pr(d_e(X^n,v^n) \leq D_e |M=m, \{\mathcal{C}_k^n\}_{k=1}^{2^{nr}}) \notag \\
&  = \sum_{j=1}^{2^{nr}} \Pr(K=j|M=m, \{\mathcal{C}_k^n\}_{k=1}^{2^{nr}}) \cdot \notag \\
& \qquad \quad \Pr(d_e(X^n,v^n) \leq D_e |M=m,K=j, \{\mathcal{C}_k^n\}_{k=1}^{2^{nr}}) \notag \\
& \leq \sum_{j=1}^{2^{nr}} 2^{-n(r-2\epsilon)} \Pr(d_e(X^n,v^n) \leq D_e |M=m,\mathcal{C}_j^n), \label{eqachfinalboundm0}
\end{align}
where the inequality follows from:
\begin{align*}
& \Pr(K=j|M=m, \{\mathcal{C}_k^n\}_{k=1}^{2^{nr}}) \\
& = \frac{ \Pr(K=j) \Pr(M=m|K=j,\{\mathcal{C}_k^n\}_{k=1}^{2^{nr}}) }{ \sum_{\ell=1}^{2^{nr}} \Pr(K=\ell) \Pr(M=m|K=\ell,\{\mathcal{C}_k^n\}_{k=1}^{2^{nr}}) } \\
& = \frac{  \Pr(M=m|K=j,\mathcal{C}_j^n) }{ \sum_{\ell=1}^{2^{nr}} \Pr(M=m|K=\ell,\mathcal{C}_\ell^n) } \\
& = \frac{ \sum\limits_{ \substack{x^n: \\ m \in C_j(x^n) }} \!\!\!\!\! \Pr(X^n \! =x^n) \Pr(M \! =m|X^n \! =x^n,K \! =j,\mathcal{C}_j^n) }{ \sum\limits_{\ell=1}^{2^{nr}} \!\!\! \sum\limits_{\substack{x^n: \\ m \in C_\ell (x^n) }} \!\!\!\!\! \Pr(X^n \! =x^n) \Pr(M \! =m|X^n\! =x^n,K \! =\ell,\mathcal{C}_\ell^n) }  \\
& \stackrel{\text{(a)}} \leq \frac{  \sum_{x^n: m \in C_j(x^n)} 1 }{ \sum_{\ell=1}^{2^{nr}} \sum_{x^n: m \in C_\ell (x^n) } 2^{-2n\epsilon} } \\
& \stackrel{\text{(b)}} = 2^{-n(r-2\epsilon)}
\end{align*} 
where (a) follows from the fact that $1 \leq N_{x^n} \leq 2^{2n\epsilon}$, and (b) follows from the fact that, for any $j$, $\sum_{x^n: m \in C_j(x^n)} 1 =  { | \{ x^n: \left(x^n,y_m^n( \mathcal{C}_j) \right) \in T_{Q_{XY}} \}|} = |T_{Q_{X|Y} ( y_m^n( \mathcal{C}_j))} | = |T_{Q_{X|Y} ( y^n)} |$ for any $y^n \in T_{Q_Y}$.

Given $\left(\cup_{k=1}^{2^{nr}} \mathcal{E}_k \right)^c$, the terms in the summands of~\eqref{eqachfinalboundm0} are independent and identically distributed random variables, with an upper bound given by~\eqref{eqachfinalboundprob}, and an expectation upper-bounded by~\eqref{eqachfinalboundexp}. It follows from Chernoff's bound~\cite[Corollary 2]{cuff2014henchman} that
\begin{align}
& \Pr \!\! \left( \!\! \Pr(d_e(X^n \! , \! v^n) \! \leq \!\! D_e |M \!\! = \!\!m, \! \{\mathcal{C}_k^n\}_{k=1}^{2^{nr}} \! ) \! > \!  2^{-n (E-8 \epsilon)} \Bigg| \! \bigcap_{k=1}^{2^{nr}} \mathcal{E}_k^c \! \right) \notag \\
& = \Pr \left( \sum_{j=1}^{2^{nr}} 2^{-n(r-2 \epsilon)} \Pr(d_e(X^n,v^n) \leq D_e |M=m,\mathcal{C}_j^n) > \right. \notag \\
& \qquad \qquad \qquad \left.  2^{-n (E-8 \epsilon)} \Bigg| \bigcap_{k=1}^{2^{nr}} \mathcal{E}_k^c \right) \notag \\
& = \Pr \left( \sum_{j=1}^{2^{nr}}  \Pr(d_e(X^n,v^n) \leq D_e |M=m,\mathcal{C}_j^n) > \right. \notag \\
& \qquad \qquad \qquad \left. 2^{-n (E-r-6 \epsilon)} \Bigg| \bigcap_{k=1}^{2^{nr}} \mathcal{E}_k^c \right) \notag \\
& \leq \left(  \frac{e2^{nr} 2^{-n(R_e(Q_X,D_e)-4 \epsilon)}  }{2^{-n ( E -r-6 \epsilon)}} \right)^{\frac{2^{-n ( E -r-6 \epsilon)}} {2^{-n(R(Q_{XY},D_e) - 4\epsilon)}} } \notag \\
& \leq \left( e 2^{-n(R_e(Q_X,D_e) - E-4\epsilon+6\epsilon)} \right)^{2^{-n(E-r-R(Q_{XY},D_e)-6 \epsilon+4\epsilon) } }  \notag \\
& \leq 2^{-\epsilon n 2^{2 \epsilon n}}, \label{eqachboundm}
\end{align}
where the last inequality follows from the fact that $R_e(Q_X,D_e)-E \geq 0$, and $E-r-R(Q_{XY},D_e) \leq 0$. By the union bound,
\begin{align}
 \Pr \left( \tilde{\mathcal{E}} \Bigg| \bigcap_{k=1}^{2^{nr}} \mathcal{E}_k^c \right)  & \leq N 2^{-\epsilon n 2^{2 \epsilon n}} \notag \\
 & \leq  2^{n(I_{Q_{XY}}(X;Y) + \epsilon - \epsilon 2^{2 \epsilon n}) } \notag \\
 & \leq  2^{n( \log |\mathcal{X}| + \epsilon - \epsilon 2^{2 \epsilon n}) } \leq 2^{-\frac{\epsilon}{2} n 2^{2 \epsilon n}}. \label{eqachboundforallm}
\end{align}
Combining~\eqref{eqachfinalboundinit} and~\eqref{eqachboundforallm} yields
\begin{align}
\Pr \left( \tilde{\mathcal{E}} \right)  \leq e^{-2^{n \epsilon/8}} + 2^{-\frac{\epsilon}{2} n 2^{2 \epsilon n}} \leq e^{-2^{n \epsilon/9}},
\end{align}
as desired. \hfill $\blacksquare$

\section{Proof of Proposition~\ref{proplimitR}} \label{AppendixproplimitR}

Consider the following proposition.
\begin{Proposition}
Given $\epsilon > 0$, $\beta >0$, and $R' > \max_{Q: D(Q||P) \leq \beta} R(Q,D)=:R_\beta$,
there exists $n_4( \epsilon, |\mathcal{X}|, |\mathcal{Y}|, R', d_e)$ such that for all $n \geq n_4$, $D \geq D_{\min}$, $D_e \geq D_{e,\min}$, and for each $Q_X \in \mathcal{Q}_{\mathcal{X}}^n(\beta,0)$ (cf.~\eqref{eqdefQXalpha}) ,
\begin{align*}
\Big| \max_{ \substack{ Q_{XY} \in \mathcal{Q}_\mathcal{XY}^n (Q_X,D): \\ I_{Q_{XY}} (X;Y) \leq R'}} R(Q_{XY},D_e) - R(Q_X,R',D,D_e) \Big| \leq \epsilon.
\end{align*}
\end{Proposition}
\begin{IEEEproof}
Note that Proposition~\ref{propoutermaxconv} is a special case in which $\beta=+\infty$ and $ R \geq \max_{Q} R(Q,D)$. As such, the proof follows along similar lines as Propositions~\ref{propoutermaxconv} and~\ref{propinnerminconv}, but must account for the rate constraint $R$.

It follows directly from the definition that
\begin{align*}
\max_{ \substack{ Q_{XY} \in \mathcal{Q}_\mathcal{XY}^n (Q_X,D): \\ I_{Q_{XY}} (X;Y) \leq R'}} R(Q_{XY},D_e) \leq R(Q_X,R',D,D_e).
\end{align*}
So, we only need to show the reverse direction. To that end, choose $R''$ such that $R_{\beta} < R'' < R'$. By Proposition~\ref{propcontinuityR}, $R(Q_X,R,D,D_e)$ is uniformly continuous in $(Q_X,R)$ over the set $\{(Q_X,R): D(Q_X||P) \leq \beta, R'' \leq R \leq R'\}$. Then let $\delta_1 > 0 $ be small enough such that, for all $Q_X \in \mathcal{Q}_{\mathcal{X}}^n(\beta,0)$, \begin{align} \label{propcontRinnereq0}
| R(Q_X,R'-\delta_1,D,D_e) - R(Q_X,R',D,D_e) | \leq \epsilon/2.
\end{align}
Let $\delta_2 > 0$ be small enough such that \begin{align} 
\| P_{XY} - P'_{XY} \| \leq \delta_2 
 \! \Rightarrow \! |R(P_{XY},D_e) - R(P'_{XY},D_e) | \leq \! \epsilon/2 \notag \\
 \text{ and } |I_{P_{XY}} (X;Y) - I_{P'_{XY}}(X;Y) | \leq \delta_1. \label{propcontRinnereq1}
\end{align}
Let $n \geq n_4 \geq |\mathcal{Y}| \sqrt{|\mathcal{X}|}/\delta_2$. Fix $Q_{X} \in \mathcal{Q}_{\mathcal{X}}^n (\beta,0)$ and let $P^\star_{Y|X}$ be the conditional distribution achieving the maximum in $R(Q_X,R'-\delta_1,D,D_e)$. We construct $P'_{Y|X}$ by rounding the values of $P^\star_{Y|X}$, as done in Proposition~\ref{propinnerminconv}. Similarly to  Proposition~\ref{propinnerminconv}, this guarantees  that $Q_X P'_{Y|X} \in \mathcal{Q}_\mathcal{XY}^n$,
\begin{align} \label{propcontRinnereq2}
\| Q_{X} P'_{Y|X} - Q_X P^\star_{Y|X} \| \! \leq \delta_2, \text{ and } \E_{Q_XP'_{Y|X}} [d(X,Y)] \leq \! D.
\end{align}
Moreover, it follows from~\eqref{propcontRinnereq1} and~\eqref{propcontRinnereq2} that $I_{Q_XP'_{XY}} (X;Y) \leq I_{Q_X P^\star_{XY}}(X;Y) + \delta_1 \leq R'$. 
Therefore, 
\begin{align*}
\max_{ \substack{ Q_{XY} \in \mathcal{Q}_\mathcal{XY}^n (Q_X,D): \\ I_{Q_{XY}} (X;Y) \leq R'}} R(Q_{XY},D_e) \vspace{-3mm}
& \geq R(Q_X P'_{Y|X},D_e) \\
& \geq R(Q_X P^\star_{Y|X},D_e) - \epsilon/2 \\
& \geq R(Q_X,R',D,D_e) - \epsilon, 
\end{align*}
where the second inequality follows from~\eqref{propcontRinnereq1} and~\eqref{propcontRinnereq2}, and the third inequality from~\eqref{propcontRinnereq0}.
\end{IEEEproof}
The proposition yields
\begin{align*}
& \min_{Q_X \in \mathcal{Q}_\mathcal{X}^n (\beta,0)} D(Q_X ||P) + R(Q,R',D,D_e) - \epsilon \\ & \leq 
\min_{Q_X \in \mathcal{Q}_\mathcal{X}^n (\beta,0)} D(Q_X ||P) + R(Q^\star_{R'}(Q_X),D_e) \\ & \leq 
 \min_{Q_X \in \mathcal{Q}_\mathcal{X}^n (\beta,0)} D(Q_X ||P) + R(Q,R',D,D_e).
\end{align*}
 By taking the limit as $n$ goes to infinity, and noting that $\epsilon$ is arbitrary, the proof is concluded.

\section{Proofs of Propositions~\ref{propcontinuity} and~\ref{propcontinuityD}} \label{propcontinuityproof}

\subsection{Proof of Proposition~\ref{propcontinuity}} \label{propcontinuityproof1}
We restate the proposition. 

\emph{Proposition~\ref{propcontinuity}}:
Let $N_1$ and $N_2$ be in $\mathbb{N}$, and let $\mathcal{S}$ and $\mathcal{U}$ be compact subsets of $\mathbb{R}^{N_1}$ and $\mathbb{R}^{N_2}$, respectively. 
Let $\nu$ be a non-negative continuous function defined on $\mathcal{S} \times \mathcal{U}$, and let $\vartheta$ be a real-valued continuous function defined on $\mathcal{S} \times \mathcal{U}$. Suppose they satisfy the following condition:
\begin{enumerate}
\item[(PA)] If  $(s,u_1) \in \mathcal{S} \times \mathcal{U}$ satisfies $\nu(s,u_1)= \min_{u' \in \mathcal{U}} \nu(s,u')$, then there exists $u_2$ such that $\vartheta(s,u_2) = \vartheta(s,u_1)$, and for all $s' \in \mathcal{S}$, $\nu(s',u_2)=\min_{u' \in \mathcal{U}} \nu(s',u').$
\end{enumerate}
Let $t_0 = \max_{s \in \mathcal{S}} \min_{u \in \mathcal{U}} \nu(s,u)$, and let $\varphi$ be a function on $\mathcal{S} \times [t_0, +\infty)$ defined as follows:
\begin{equation*}
\varphi(s,t) = \min_{u: \nu(s,u) \leq t} \vartheta (s,u).
\end{equation*}
If for fixed $s \in \mathcal{S}$, $\varphi(s,t)$ is continuous in $t$, then $\varphi(s,t)$ is continuous in the pair $(s,t)$.
\bigskip

First, note that, for all $s \in \mathcal{S}$ and all $t \geq t_0$, $\nu^{-1}\left(s,[0,t] \right) \triangleq \{ u: \nu(s,u) \leq t\}$ is closed by continuity of $\nu$, so it is compact since it is also bounded. Moreover it is non-empty since $t \geq t_0$. Since $\vartheta$ is continuous and the minimization is over a compact set, $\varphi$ is well defined.

Now fix $(s,t) \in \mathcal{S} \times [t_0, +\infty)$, and consider any sequence $(s_k,t_k) \rightarrow (s,t)$. Let $t_s = \min_{u \in \mathcal{U}} \nu(s,u)$ and consider any $\epsilon > 0$. 

If $t > t_s$: 

By continuity of $\varphi(s,t)$ as a function of $t$ for fixed $s$, there exists $\delta >0$ such that $ |t - t'| \leq \delta \Rightarrow |\varphi(s,t) - \varphi(s,t')| \leq \epsilon$. Let $t'=t- \min\{ \delta/2, (t-t_s)/2 \} $, and let $u' \in \argmin_{u: \nu(s,u) \leq t'} \vartheta(s,u)$. Then, $\nu(s,u') < t$ and $\varphi(s,t') = \vartheta (s,u') \leq \varphi(s,t)+\epsilon$. \\

If $t=t_s$:

Let $u'$ be a minimizer for $\varphi(s,t_s)$ satisfying $\nu(s',u')=t_{s'}$ for all $s' \in \mathcal{S}$. Such choice is possible by assumption (PA). Note that $\vartheta(s,u')=\varphi(s,t_s).$ \\

We claim that the choice of $u'$ is feasible for the minimization in $\varphi(s_k,t_k)$, i.e., $\nu(s_k,u')\leq t_k$ for sufficiently large $k$. Indeed, if $t>t_s$, $\nu(s_k,u') \rightarrow\nu(s,u') = t' < t$, then for sufficiently large $k$, $\nu(s_k,u') \leq t_k$.  If $t=t_s$, then $\nu(s_k,u')=t_{s_k} \leq t_0 \leq t_k$. 

Moreover, by continuity of $\vartheta$, $\vartheta(s_k,u') \rightarrow \vartheta (s,u')$. Then, for sufficiently large $k$, $\vartheta(s_k,u') \leq \varphi(s,t) + \epsilon/2.$ So, we get
\begin{equation*}
\limsup_{k \rightarrow \infty}  \varphi(s_k,t_k) \leq \limsup_{k \rightarrow \infty}  \vartheta(s_k,u') \leq  \varphi(s,t).
\end{equation*}
On the other hand, let $u_k$ be a minimizer for $\varphi(s_k,t_k)$. Consider a sequence of integers $\{k_j\}$ such that 
\begin{equation*}
\varphi(s_{k_j}, t_{k_j}) \rightarrow \liminf_{k \rightarrow \infty} \varphi(s_k,t_k).
\end{equation*}
Let $\{u_{k_j} \}$ be the corresponding subsequence of minimizers. Since $\mathcal{U}$ is a bounded set, then $\{u_{k_j} \}$ has a convergent subsequence $\{u_{k_{j_\ell}} \}.$ Let $u'$ be its limit. By continuity of $\nu$, we have $\nu(s,u') = \lim_{\ell \rightarrow \infty} \nu(s_{k_{j_\ell}},u_{k_{j_\ell}}) \leq \lim_{\ell \rightarrow \infty} t_{k_{j_\ell}} =t.$ Therefore,
\begin{align*}
\varphi(s,t) \leq \vartheta(s,u') & = \lim_{\ell \rightarrow \infty} \vartheta (s_{k_{j_\ell}},u_{k_{j_\ell}}) \\
& = \lim_{\ell \rightarrow \infty} \varphi (s_{k_{j_\ell}},t_{k_{j_\ell}}) \\
& = \liminf_{k \rightarrow \infty} \varphi(s_k,t_k). 
\end{align*} \hfill $\blacksquare$

\subsection{Proof of Proposition~\ref{propcontinuityD}}
\label{propcontinuityproof2}
We restate the proposition.

\emph{Proposition~\ref{propcontinuityD}}:
Let $N$ be in $\mathbb{N}$, and let $\mathcal{T}$ be a non-empty compact subset of $\mathbb{R}^{N}$. Let $L$ be a real-valued continuous function defined on $\mathcal{T}$. Let $T_1 \supseteq T_2 \supseteq \cdots$ be a decreasing sequence of non-empty compact subsets of $\mathcal{T}$. Let $T= \bigcap_{i \geq 1} T_i$. Then,
\begin{equation*}
\lim_{k \rightarrow \infty} \max_{t \in T_k} L(t) = \max_{t \in T} L(t).       
\end{equation*} 
Moreover, let $S_1 \subseteq S_2 \subseteq \cdots$ be an increasing sequence of non-empty compact subsets of $\mathcal{T}$. Let $S = \overline{\bigcup_{i \geq 1} S_i}$ (where the bar denotes closure of the set). Then,
\begin{equation*}
\lim_{k \rightarrow \infty} \max_{t \in S_k} L(t) = \max_{t \in S} L(t).        
\end{equation*} 
Consequently, if $\mathcal{T}$ is also convex, and $L_c$ is a real-valued convex and continuous function defined on $\mathcal{T}$ with $s_0 = \min_{t \in \mathcal{T}} L_c(t)$, then
\begin{align*}
\hat{L}(s) := \max_{t: L_c(t) \leq s} L(t)
\end{align*}
is continuous in $s \in [s_0, +\infty)$. \\

\indent First, note that $T$ is non-empty and compact since a countable intersection of non-empty  decreasing compact sets is non-empty and compact. Let 
\begin{equation*}
t_k = \argmax_{t \in T_k} L(t) \quad \text{ and } \quad t^\star = \argmax_{t \in T} L(t). 
\end{equation*}
We need to show that $L(t_k) \rightarrow L(t^\star)$. Let $\mathcal{B}_\delta (t) = \{ t' \in \mathcal{T}: \|t'-t \| < \delta\}$, and
consider the following claim. 

\emph{Claim 1:} For all $\delta >0$, there exists $k_0$ such that for all $k \geq k_0$, $T_k \subseteq \mathcal{B}_\delta(T)$, where 
\begin{equation*}
\mathcal{B}_\delta(T) = \bigcup_{t \in T} \mathcal{B}_\delta (t).
\end{equation*}
We show first how the claim yields our result. Let $\epsilon > 0$ be given. By the uniform continuity of $L$ (continuity on a compact set), there exists $\delta > 0$ such that $\| t - t' \| \leq \delta \Rightarrow | L(t)-L(t') | \leq \epsilon. $ Let $k$ be large enough as guaranteed by the claim. Then, for all $t \in T_k$, there exists $t' \in T$ such that $\| t - t'\| \leq \delta$, and subsequently $| L(t) -L(t') | \leq \epsilon$. In particular, there exists $t' \in T$ such that $| L(t_k) - L(t') | \leq \epsilon$. Then, we get $L(t_k) \leq L(t') + \epsilon \leq L(t^\star) + \epsilon$. Since $L(t_k) \geq L(t^\star)$, we get $ | L(t_k) - L(t^\star)| \leq \epsilon$. Therefore, $L(t_k) \rightarrow L(t^\star)$. It remains to prove the claim to establish the first part of the proposition.

\emph{Proof of Claim 1:}  Fix $\delta >0$.  $\mathcal{B}_\delta(T)$ is open in $\mathcal{T}$ by construction. Therefore, $T_k \backslash \mathcal{B}_\delta(T)$ is closed in $\mathcal{T}$. Since $\mathcal{T}$ is closed in $\mathbb{R}^{N}$, then $T_k\backslash \mathcal{B}_\delta(T)$ is also closed in $\mathbb{R}^{N}$. Moreover, it is bounded, so it is compact. Since
\begin{equation*}
\bigcap_{i \geq 1} T_k\backslash \mathcal{B}_\delta(T) = \left ( \bigcap_{i \geq 1} T_k \right) \Big\backslash \mathcal{B}_\delta(T) = T \backslash \mathcal{B}_\delta(T) = \varnothing,
\end{equation*} 
and $T_k\backslash \mathcal{B}_\delta(T)$ is a decreasing sequence of compact sets, there exists $k_0$ such that for all $k \geq k_0$, $T_k \backslash \mathcal{B}_\delta(T)$ is empty. \\

Similarly, to prove the second part of the proposition, let 
\begin{equation*}
s_k = \argmax_{t \in S_k} L(t) \quad \text{ and } \quad s^\star = \argmax_{t \in S} L(t). 
\end{equation*}
We need to show that $L(s_k) \rightarrow L(s^\star)$. To this end, consider the following claim. 

\emph{Claim 2:} For all $\delta >0$, there exists $k_1$ such that for all $k \geq k_1$, $  S \subseteq \mathcal{B}_\delta(S_k) $.\\
We show first how the claim yields our result. Let $\epsilon > 0$ be given. By the uniform continuity of $L$, there exists $\delta > 0$ such that $\| t - t' \| \leq \delta \Rightarrow | L(t)-L(t') | \leq \epsilon. $ Let $k$ be large enough as guaranteed by the claim. Then, for all $t \in S$, there exists $t' \in S_k$ such that $\| t - t'\| \leq \delta$, and subsequently $| L(t) -L(t') | \leq \epsilon$. In particular, there exists $t' \in S_k$ such that $| L(s^\star) - L(t') | \leq \epsilon$. Then, we get $L(s_k) \geq L(t') \geq L(s^\star) - \epsilon$. Since $L(s_k) \leq L(s^\star)$, we get $ | L(s_k) - L(s^\star)| \leq \epsilon$. Therefore, $L(s_k) \rightarrow L(s^\star)$. It remains to prove the claim.

\emph{Proof of Claim 2:}  Fix $\delta >0$.  $\mathcal{B}_\delta(S_k)$ is open in $\mathcal{T}$ by construction. Therefore, $S \backslash \mathcal{B}_\delta(S_k)$ is closed in $\mathcal{T}$. Then $S \backslash \mathcal{B}_\delta(S_k)$ is closed in $\mathbb{R}^{N}$. Moreover, it is bounded, so it is compact. Since
\begin{equation*}
\bigcap_{i \geq 1} S \backslash \mathcal{B}_\delta(S_k) \! = \! S \Big \backslash \! \left( \! \bigcup_{i \geq 1} \mathcal{B}_\delta(S_k) \! \right) \!\! = \! \overline{\bigcup_{i \geq 1} S_i} \Big \backslash \mathcal{B}_\delta \! \left( \! \bigcup_{i \geq 1} S_i  \! \right) \! = \! \varnothing,
\end{equation*} 
and $S \backslash \mathcal{B}_\delta(S_k)$ is a decreasing sequence of compact sets, there exists $k_1$ such that for all $k \geq k_1$, $S \backslash \mathcal{B}_\delta(S_k)$ is empty.  \\

Finally, consider $\hat{L}(s)$. If $L_c$ is a constant function, then the statement is trivial. If not, consider $s \geq s_0$, and let $s_k$ be a decreasing sequence converging to $s$. Then, 
\begin{align*}
\lim_{k \rightarrow \infty} \hat{L}(s_k) = \lim_{k \rightarrow \infty} \max_{t: L_c(t) \leq s_k} L(t) =  \max_{t: L_c(t) \leq s} L(t) = \hat{L}(s),
\end{align*}
where the second equality follows from the first part of the proposition. Therefore, $\hat{L}(s)$ is right-continuous. Now, consider $s > s_0$, and let $s_k$ be an increasing sequence converging to $s$. Note that, 
\begin{align*}
\bigcup_{k \geq 1} \left\lbrace t \in \mathcal{T}: L_c(t) \leq s_k \right\rbrace = \left\lbrace t \in \mathcal{T}: L_c(t) < s \right\rbrace.
\end{align*} 
Denote the above set by $S^-$ and let $S = \{  t \in \mathcal{T}: L_c(t) \leq s \}.$ The second part of the proposition implies that
\begin{align*}
\lim_{k \rightarrow \infty} \hat{L}(s_k) = \lim_{k \rightarrow \infty} \max_{t: L_c(t) \leq s_k} L(t) = \max_{t \in \overline{S^-}} L(t).
\end{align*}
So it suffices to show that $\overline{S^-}=S$. Clearly, $\overline{S^-} \subseteq S$ since $S$ is closed and $S^- \subseteq S$. It remains to show that any point $\tilde{t}$ satisfying $L_c(\tilde{t}) = s$ is a boundary point of $S^-$. To that end, note that $L_c(\tilde{t})$ is not a local minimum since $L_c(\tilde{t}) = s > s_0$ and $L_c$ is convex by assumption. Therefore, any neighborhood of $\tilde{t}$ intersects $S^-$. As such $\overline{S^-}=S$, and  $\hat{L}(s)$ is left-continuous, as desired. \hfill $\blacksquare$

\end{appendices}

\section*{Acknowledgment}
We are grateful to the anonymous reviewers for the helpful
  comments, especially the reviewer who pointed out the connection
  between Schieler and Cuff's analysis and our setup. This research was supported by the US National Science Foundation under grants 10-65352 and 12-18578.

\bibliographystyle{IEEEtran}
\bibliography{IEEEabrv,database}

\begin{thebibliography}{10}
\providecommand{\url}[1]{#1}
\csname url@samestyle\endcsname
\providecommand{\newblock}{\relax}
\providecommand{\bibinfo}[2]{#2}
\providecommand{\BIBentrySTDinterwordspacing}{\spaceskip=0pt\relax}
\providecommand{\BIBentryALTinterwordstretchfactor}{4}
\providecommand{\BIBentryALTinterwordspacing}{\spaceskip=\fontdimen2\font plus
\BIBentryALTinterwordstretchfactor\fontdimen3\font minus
  \fontdimen4\font\relax}
\providecommand{\BIBforeignlanguage}[2]{{%
\expandafter\ifx\csname l@#1\endcsname\relax
\typeout{** WARNING: IEEEtran.bst: No hyphenation pattern has been}%
\typeout{** loaded for the language `#1'. Using the pattern for}%
\typeout{** the default language instead.}%
\else
\language=\csname l@#1\endcsname
\fi
#2}}
\providecommand{\BIBdecl}{\relax}
\BIBdecl

\bibitem{SSHTiming}
D.~X. Song, D.~Wagner, and X.~Tian, ``Timing analysis of keystrokes and timing
  attacks on {SSH},'' in \emph{Proceedings of the 10th USENIX Security
  Symposium}.\hskip 1em plus 0.5em minus 0.4em\relax Berkeley, CA, USA: USENIX
  Association, 2001.

\bibitem{PeepingTom}
\BIBentryALTinterwordspacing
K.~Zhang and X.~Wang, ``Peeping tom in the neighborhood: Keystroke
  eavesdropping on multi-user systems,'' in \emph{Proceedings of the 18th
  USENIX Security Symposium}.\hskip 1em plus 0.5em minus 0.4em\relax Montreal,
  Canada: USENIX, 2009. [Online]. Available: \url{https://www.usenix.org/node/}
\BIBentrySTDinterwordspacing

\bibitem{ParvAnonymous}
P.~Venkitasubramaniam, T.~He, and L.~Tong, ``Anonymous networking amidst
  eavesdroppers,'' \emph{IEEE Trans. Inf. Theory}, vol.~54, no.~6, pp.
  2770--2784, June 2008.

\bibitem{ParvAnonymityChaum}
P.~Venkitasubramaniam and A.~Mishra, ``Anonymity of memory-limited {C}haum
  mixes under timing analysis: An information theoretic perspective,''
  \emph{IEEE Trans. Inf. Theory}, vol.~61, no.~2, pp. 996--1009, Feb 2015.

\bibitem{SharedSchedulers}
S.~Kadloor, N.~Kiyavash, and P.~Venkitasubramaniam, ``Mitigating timing side
  channel in shared schedulers,'' \emph{Networking, IEEE/ACM Transactions on},
  vol.~PP, no.~99, pp. 1--12, 2015.

\bibitem{TimingSuh}
Y.~Wang and G.~E. Suh, ``Efficient timing channel protection for on-chip
  networks,'' in \emph{Networks on Chip (NoCS), Sixth IEEE/ACM International
  Symposium on}, May 2012, pp. 142--151.

\bibitem{DiffPrivacyMech}
F.~McSherry and K.~Talwar, ``Mechanism design via differential privacy,'' in
  \emph{Foundations of Computer Science, 48th Annual IEEE Symposium on}, Oct
  2007, pp. 94--103.

\bibitem{DiffPrivacySurvey}
C.~Dwork, ``Differential privacy: A survey of results,'' in \emph{Theory and
  applications of models of computation}.\hskip 1em plus 0.5em minus
  0.4em\relax Springer, 2008, pp. 1--19.

\bibitem{ShannonSecrecy}
C.~E. Shannon, ``Communication theory of secrecy systems,'' \emph{Bell System
  Technical Journal}, vol.~28, no.~4, pp. 656--715, 1949.

\bibitem{merhav2015probmetric}
\BIBentryALTinterwordspacing
N.~Weinberger and N.~Merhav, ``A large deviations approach to secure lossy
  compression,'' May 2015. [Online]. Available:
  \url{http://arxiv.org/abs/1504.05756}
\BIBentrySTDinterwordspacing

\bibitem{merhavsecurelossyTrans}
------, ``A large deviations approach to secure lossy compression,'' \emph{IEEE
  Trans. Inf. Theory}, vol.~PP, no.~99, pp. 1--1, 2017.

\bibitem{WireTapChannel}
A.~D. Wyner, ``The wire-tap channel,'' \emph{Bell System Technical Journal},
  vol.~54, no.~8, pp. 1355--1387, Oct. 1975.

\bibitem{SecrecyFading}
P.~K. Gopala, L.~Lai, and H.~El~Gamal, ``On the secrecy capacity of fading
  channels,'' \emph{IEEE Trans. Inf. Theory}, vol.~54, no.~10, pp. 4687--4698,
  Oct. 2008.

\bibitem{relayeavesdropper}
L.~Lai and H.~El~Gamal, ``The relay-eavesdropper channel: cooperation for
  secrecy,'' \emph{IEEE Trans. Inf. Theory}, vol.~54, no.~9, pp. 4005--4019,
  Sept 2008.

\bibitem{relayhelpereve}
N.~Marina, H.~Yagi, and H.~V. Poor, ``Improved rate-equivocation regions for
  secure cooperative communication,'' in \emph{Proc. IEEE Int. Symp. Inf.
  Theory (ISIT)}, July 2011, pp. 2871--2875.

\bibitem{csiszar1978broadcast}
I.~Csisz{\'a}r and J.~K{\"o}rner, ``Broadcast channels with confidential
  messages,'' \emph{IEEE Trans. Inf. Theory}, vol.~24, no.~3, pp. 339--348, May
  1978.

\bibitem{merhav2008shannon}
N.~Merhav, ``Shannon's secrecy system with informed receivers and its
  application to systematic coding for wiretapped channels,'' \emph{IEEE Trans.
  Inf. Theory}, vol.~54, no.~6, pp. 2723--2734, June 2008.

\bibitem{erkippoor2008lossless}
D.~Gunduz, E.~Erkip, and H.~V. Poor, ``Lossless compression with security
  constraints,'' in \emph{Proc. IEEE Int. Symp. Inf. Theory (ISIT)}, July 2008,
  pp. 111--115.

\bibitem{GuessandEntropy}
J.~Massey, ``Guessing and entropy,'' in \emph{Proc. IEEE Int. Symp. Inf. Theory
  (ISIT)}, Jun 1994, pp. 204--.

\bibitem{MerhavArikanShannonCipher}
N.~Merhav and E.~Ar{\i}kan, ``The {S}hannon cipher system with a guessing
  wiretapper,'' \emph{IEEE Trans. Inf. Theory}, vol.~45, no.~6, pp. 1860--1866,
  1999.

\bibitem{yamamoto1997shannoncipher}
H.~Yamamoto, ``Rate-distortion theory for the {S}hannon cipher system,''
  \emph{IEEE Trans. Inf. Theory}, vol.~43, no.~3, pp. 827--835, 1997.

\bibitem{yamamoto1988secdec}
------, ``A rate-distortion problem for a communication system with a secondary
  decoder to be hindered,'' \emph{IEEE Trans. Inf. Theory}, vol.~34, no.~4, pp.
  835--842, 1988.

\bibitem{cuff2013secrecycausal}
C.~Schieler and P.~Cuff, ``Rate-distortion theory for secrecy systems,''
  \emph{IEEE Trans. Inf. Theory}, vol.~60, no.~12, pp. 7584--7605, Dec 2014.

\bibitem{cuffisit2014henchman}
------, ``The henchman problem: measuring secrecy by the minimum distortion in
  a list,'' in \emph{Proc. IEEE Int. Symp. Inf. Theory (ISIT)}.\hskip 1em plus
  0.5em minus 0.4em\relax IEEE, 2014, pp. 596--600.

\bibitem{cuff2014henchman}
------, ``The henchman problem: Measuring secrecy by the minimum distortion in
  a list,'' \emph{IEEE Trans. Inf. Theory}, vol.~62, no.~6, pp. 3436--3450,
  June 2016.

\bibitem{merhav2003perfectsecrecy}
N.~Merhav, ``A large-deviations notion of perfect secrecy,'' \emph{IEEE Trans.
  Inf. Theory}, vol.~49, no.~2, pp. 506--508, 2003.

\bibitem{arikanguessingdistorion}
E.~Ar{\i}kan and N.~Merhav, ``Guessing subject to distortion,'' \emph{IEEE
  Trans. Inf. Theory}, vol.~44, no.~3, pp. 1041--1056, May 1998.

\bibitem{korner}
I.~Csisz{\'a}r and J.~K{\"o}rner, \emph{Information theory: coding theorems for
  discrete memoryless systems}.\hskip 1em plus 0.5em minus 0.4em\relax Academic
  Press, New York, 1981.

\bibitem{SuccessiveRefinement}
W.~H.~R. Equitz and T.~M. Cover, ``Successive refinement of information,''
  \emph{IEEE Trans. Inf. Theory}, vol.~37, no.~2, pp. 269--275, Mar 1991.

\bibitem{WeissmanEmpiricalRateConstrained}
T.~Weissman and E.~Ordentlich, ``The empirical distribution of rate-constrained
  source codes,'' \emph{IEEE Trans. Inf. Theory}, vol.~51, no.~11, pp.
  3718--3733, Nov 2005.

\end{thebibliography}

\begin{IEEEbiographynophoto}
{Ibrahim Issa} received his B.E. degree in Computer and Communications
Engineering from the American University of Beirut, Lebanon, in 2012.
He is currently pursuing his Ph.D in Electrical and Computer Engineering
at Cornell, Ithaca, NY. His research interests include information-theoretic
security and quantum information theory.
He was the recipient of the Jacobs fellowship for the academic
year 2012-2013.
\end{IEEEbiographynophoto}
\begin{IEEEbiographynophoto}
{Aaron B. Wagner} (S'13-M'05-SM'00) received the B.S. degree in Electrical
Engineering from the University of Michigan, Ann Arbor, in 1999 and the
M.S. and Ph.D. degrees in Electrical Engineering and Computer Sciences from
the University of California, Berkeley, in 2002 and 2005, respectively. During
the 2005-2006 academic year, he was a Postdoctoral Research Associate in
the Coordinated Science Laboratory at the University of Illinois at Urbana-Champaign
and a Visiting Assistant Professor in the School of Electrical and
Computer Engineering at Cornell University. Since 2006, he has been with
the School of Electrical and Computer Engineering at Cornell, where he is
currently an Associate Professor.

He has received the NSF CAREER award, the David J. Sakrison Memorial
Prize from the U.C. Berkeley EECS Dept., the Bernard Friedman Memorial
Prize in Applied Mathematics from the U.C. Berkeley Dept. of Mathematics,
and teaching awards at the department, college, and university level at Cornell.
\end{IEEEbiographynophoto}

\end{document}